\documentclass[conference]{IEEEtran}
\usepackage{graphicx} 

\title{Distributing Context-Aware Shared Memory Data Structures: A Case Study on Singly-Linked Lists}

\date{April 2023}
\usepackage{xspace}
\usepackage{subcaption}
\usepackage{comment}
\usepackage{color}
\usepackage{cite}
\usepackage[ruled, linesnumbered]{algorithm2e}
\usepackage{paralist}
\usepackage{enumitem}
\usepackage{amsthm}
\usepackage{mathtools}
\usepackage{cleveref}

\AtBeginEnvironment{appendices}{\crefalias{section}{appendix}}

\newif\iffullalg\fullalgfalse

\crefname{algocf}{alg.}{algs.}
\Crefname{algocf}{Algorithm}{Algorithms}
\SetKwProg{Fn}{}{:}{end}
\SetKwProg{Struct}{struct}{ \{}{\}\;}
\SetKwComment{Comment}{\slash \slash }{}

\definecolor{asparagus}{rgb}{0.53, 0.66, 0.42}

\newcommand{\sandeepnote}[1]{{\color{asparagus}#1}}
\newcommand{\raaghavnote}[1]{{\color{magenta}#1}}
\definecolor{highlight}{rgb}{0.54, 0.17, 0.89}

\newcommand{\hl}[1]{{\color{highlight}#1}}

\newcommand{\br}[1]{\langle #1 \rangle\xspace}

\newcommand{\Next}{{\tt Next\xspace}\xspace}

\newcommand{\Lookup}{{\tt Lookup}\xspace}
\newcommand{\Head}{{\tt Head}\xspace}

\newcommand{\dummyNode}{{\tt DummyNode}\xspace}
\newcommand{\Tail}{{\tt Tail}\xspace}

\newcommand{\subtail}{{\tt SubTail\xspace}\xspace}
\newcommand{\sublist}{{sublist}\xspace}
\newcommand{\sublists}{{sublists}\xspace}
\newcommand{\Registry}{{\tt Registry\xspace}\xspace}

\newcommand{\subhead}{{\tt SubHead\xspace}\xspace}
\newcommand{\subheads}{{\tt SubHeads\xspace}\xspace}

\newcommand{\AtomicInteger}{{\tt AtomicInteger}\xspace}
\newcommand{\ItemRef}{{\tt ItemRef\xspace}\xspace}
\newcommand{\Item}{{\tt Item\xspace}\xspace}
\newcommand{\GetItem}{{\tt GetItem\xspace}\xspace}

\newcommand{\InsertAfter}{{\tt{InsertAfter}}\xspace}
\newcommand{\Insert}{{\tt Insert}\xspace}
\newcommand{\Search}{{\tt Search}\xspace}

\newcommand{\Key}{{\tt Key}\xspace}
\newcommand{\Delete}{{\tt{Delete}}\xspace}
\newcommand{\Split}{{\tt Split}\xspace}
\newcommand{\Move}{{\tt Move}\xspace}
\newcommand{\SwitchLB}{{\tt Switch\xspace}\xspace}
\newcommand{\newLocation}{{newLocation}\xspace}
\newcommand{\PArrow}{{\rightarrow}}

\newcommand{\AbortingMoveMethod}{{\tt AM}\xspace}
\newcommand{\TempReplMethod}{{\tt TR}\xspace}

\newcommand{\StartCount}{{startCount}\xspace}
\newcommand{\newStartCount}{{newStartCount}\xspace}
\newcommand{\newEndCount}{{newEndCount}\xspace}
\newcommand{\EndCount}{{endCount}\xspace}
\newcommand{\offset}{{offset}\xspace}
\newcommand{\offsets}{{offsets}\xspace}
\newcommand{\counterTemp}{{counterTemp}\xspace}
\newcommand{\TS}{{TS}\xspace}
\newcommand{\LC}{{LC}\xspace}
\newcommand{\Replicate}{{\tt Replicate}\xspace}
\newcommand{\ReplicateInsertAfter}{{\tt ReplicateInsertAfter}\xspace}
\newcommand{\ReplicateDelete}{{\tt ReplicateDelete}\xspace}

\theoremstyle{definition}
\newtheorem{definition}{Definition}[section]

\newtheorem{theorem}{Theorem}[section]

\newtheorem{lemma}[theorem]{Lemma}

\begin{document}

\author{\IEEEauthorblockN{Raaghav Ravishankar$^*$, Sandeep Kulkarni$^*$, Sathya Peri$^\dag$, Gokarna Sharma$^\ddag$}
\IEEEauthorblockA{
$^*$Michigan State University, East Lansing, MI 48824, USA,
{\it \{ravisha7, sandeep\}@msu.edu}}

$^\dag$IIT Hyderabad, Telangana, India,
{\it sathya\_p@cse.iith.ac.in}\\

 $^\ddag$Kent State University, Kent, OH 44242, USA,
{\it gsharma2@kent.edu}
}

\maketitle

\begin{abstract}
In this paper, we study the partitioning of a \textit{context-aware} shared memory data structure so that it can be implemented as a distributed data structure running on multiple machines. By context-aware data structures, we mean that the result of an operation not only depends upon the \textit{value of the shared data} but also upon the previous operations performed by the same client. 
%
While there is substantial work on 
designing 
distributed data structures, designing 
distributed context-aware data structures has not received much attention. 

We focus on singly-linked lists
as a case study of the context-aware data structure. We start with a shared memory context-aware lock-free 
singly-linked list and show how it can be 
{\textit{transformed} into} a distributed lock-free context-aware
singly-linked list. 
The main challenge in such a transformation is to preserve properties of client-visible operations of the underlying data structure.
We present two protocols that preserve these properties of client-visible operations of the linked list. In the first protocol, the distribution is done in the background as a low priority task, while in the second protocol the client-visible operations \textit{help} the task of distribution without affecting client latency. 
In both protocols, the client-visible operations remain lock-free.
Also, our transformation approach does not utilize any hardware primitives (except a compare-and-swap operation on a single word). 
We note that our transformation is generic and can be used for other lock-free context-aware data structures
that can be constructed from singly-linked lists.
\end{abstract}

\section{Introduction} \label{sec:introduction}
Historically, it is assumed that a data structure (e.g., a hash table, list, queue, stack) is being modified by one operation (e.g., insert, delete, enqueue, dequeue) at a time. This limited the throughput that one can achieve as the operations had to be performed sequentially. To overcome this limitation, we focused on concurrent data structures where the operations could be performed simultaneously as long as they are being done in different parts (e.g., different parts of the list). And, if concurrent operations conflicted with each other then one or more of them would be aborted and retried.

Concurrent data structures are limited by the computing (and possibly communication) ability of the computer/node hosting them.
One way to overcome these limitations is to distribute them across several machines. That would allow multiple operations to be performed simultaneously. The most common example in this context is distributed hash tables \cite{DHT} (or key value stores \cite{cassandra}) where a single hash table (respectively, key-value store) is split into multiple partitions and each partition is handled by a different node. This increases
the throughput of the hash table as each node can work in parallel.

A typical hash table implementation \cite{cormen2022introduction} provides Insert, Delete and Lookup operations. A key characteristic of these operations is that, given the content of the hash table, we can uniquely determine the output of these functions. For example, if the hash table contains the keys \{1, 2, 3\} and we invoke Delete(2), the resulting hash table will then be \{1, 3\}. Invoking Lookup(1) will return the value True. In other words, the output of hash table operations depends entirely on the state of the hash table when the operation is invoked and not on previous operations performed by a given client on it.

This property of hash tables is not satisfied by many data structures. As an example, consider the unordered list data structure that supports Lookup, Next and Insert operations. An unordered list may contain duplicate keys. If the list is of the form, $\br{1, 5, 2}$ and we invoke Lookup(5), it will return a reference to the element 5. If we invoke Next operation on it, it will return a reference to the element 2. Furthermore, the code ‘Lookup(5), Insert(4)’ will result in the list to be $\br{1, 5, 4, 2}$. On the other hand, the code ‘Lookup(1), Insert(4)’ will result in $\br{1, 4, 5, 2}$. Here the result of ‘Insert(4)’ changes depending upon whether the previous operation by that client was ‘Lookup(5)’ or ‘Lookup(1)’. Likewise, the result of next operation depends upon previous operations by the same client. In other words, the result of an operation not only depends upon the state of the list but also on previous operations performed by the same client. We call such data structures as being \textit{context-aware}.

Such context-aware data structures occur in practice in many cases. As an example, consider the use case of a large music playlist or a large photograph (photo) album. A group of mobile phone users (such as a family) have made a common album but are willing to commit only a limited amount of storage for the album. They have their devices connected through a private network to look at the album whenever needed. They can do a $\Lookup(key)$ to fetch a specific photo matching a given key. Here key could be a characteristic of the picture (e.g., portrait, nature photography) thereby making it necessary that several elements with the same key are present in the list.  Since an album is ordered with one photo after another, they may want to do a $\InsertAfter(node, key)$ to insert a new photo after a specific photo in the album. They may also want to skim through the album one photo at a time by a traversal with $\Next$ operation. If a specific photo doesn't fit a given place, they can $\Delete$ the photo, and $\Insert$ it at another place as they skim through using subsequent $\Next$ operations. Furthermore, it is desirable that if two users are traversing different parts of the list, or inserting elements in different parts of the list, then their operations are unaffected by each other. This could be achieved if we \textit{split} the list across multiple machines such that each machine only maintains a part of the list. In this case, if two users are traversing through different parts of the list, their code is likely to be executed on different machines. Furthermore, when a client is traversing through the list, it will benefit from \textit{locality of reference}; if it finds one of the elements on a particular server, the subsequent elements in the traversal are likely to be on the same server as well. This means that there is a potential to perform more complex operations such as `get the next $k$ items’, if desired.

\textbf{Contributions of the paper: }
We focus on the question of \textit{transforming} a concurrent 
shared memory algorithm for a context-aware data structure into a distributed protocol so that the data structure can be split into multiple machines. The protocol is executed by 
a set of machines (or \textit{servers}) that can only communicate via message-passing. We focus our study on a singly-linked list. Specifically, we present a  protocol for transforming a concurrent shared memory implementation of an unordered linked list and a sorted linked list into distributed ones. 
We separate the linked list operations into two categories: {\em client-visible} operations that clients perform on the list and {\em client-invisible} operations that a system performs in the background.  It is simple to achieve the transformation either by weakening the guarantees of client-visible operations or, by blocking all write operations,
while a part of the list is being \textit{moved} from one server to another, for load balancing. 
However, it is challenging to achieve the transformation without losing both the client guarantees of the shared memory list, and the availability of client-visible operations.


To address this challenge, we illustrate a lock-free context-aware list with $\Next$, $\Lookup$, $\InsertAfter$ and $\Delete$ as its client-visible operations, and present two protocols that transform the list in a way that the client-visible operations are never blocked, and that their operation semantics are preserved
even when a part of the list is being moved to another server. 
%
Our protocols only rely on compare-and-swap (CAS) instructions (and AtomicInteger/RDCSS \cite{MultiCAS} that can be built from CAS) that are supported by most modern processors including Intel 64 and AMD64.
To our knowledge, this is the first such transformation for context-aware data structures.

To achieve successful transformation we introduce the notion of \sublists within a list. We add operations to \textit{split} a given \sublist into multiple \sublists and on an as-needed basis, \textit{move} a \sublist to another server, and thereby, \textit{switch} its server ownership. We also ensure that these client-invisible operations (e.g., garbage collection, load balancing, etc.) do not interfere with client-visible operations (e.g., insert, delete, etc.) on the list. 


\textbf{Outline of the paper: }\label{subsec:Outline} 
    The rest of the paper is organized as follows. 
    \Cref{sec:lockfreeconclist} first presents a concurrent, shared memory, unordered linked list and its API. Then \Cref{sec:FrameworkForDistribution} describes an approach for distributing the linked list, and 
    discusses the requirements of a transformation protocol, by defining a set of  operations to \textit{support} the transformation. 
    We then present two transformation protocols: Aborting Move (\AbortingMoveMethod) in \Cref{sec:AMSection}, and Temporary Replication (\TempReplMethod) in \Cref{sec:TRSection}. 
    Due to space constraints, their proofs of correctness are provided in Appendix sections \ref{AMAppendix} and \ref{TRAppendix}.
    \Cref{sec:Discussion} compares the protocols and provides ways to further augment the linked list without losing the ability to support the transformation. Related work is then discussed in \Cref{sec:RelatedWork}. Finally, Section \ref{sec:Conclusion} concludes the paper with immediate extensions and a short discussion on future work. 
\section{A lock-free shared memory list}\label{sec:lockfreeconclist}

\begin{figure}
    \centering
    \includegraphics[scale=0.7]{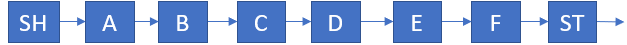}
    \caption{A sample list} 
\label{fig:sampleList}  
\vspace{-3mm}
\end{figure}
In this section, we recall a simple implementation of a list. \Cref{fig:sampleList} shows a sample list, with a head node SH, a tail node ST and 6 nodes inside it. We refer to Key type as the value used to $\Lookup$ for any node (or item). The node can additionally be equipped to store payload data ({such as a photo}). {To mark some special nodes such as sentinel nodes, Key reserves a few values such as the $\Head$ and $\Tail$.}  
 We first provide its API semantics, then justify its context-awareness and finally illustrate its shared-memory implementation.

\subsection{API for Distributed Linked List}\label{subsec:API}

\begin{itemize}[nosep, leftmargin=*]
\item $\ItemRef$ \ $\Head$()\ /\ $\Tail$(): The $\Head$() (respectively, $\Tail$())  provides a reference to the first (respectively, last) item in the list. These are immutable dummy items. \footnote{Since the implementations of \ $\Head$, $\Tail$ and $\GetItem$ are straightforward, we omit them in further discussion.}

\item $(\ItemRef, \Item)\ \GetItem\ (\ItemRef)$: The $\GetItem$ operation takes a reference to a list node as parameter, and returns the data structure of the node contained within it.
Also, if the item has moved since $\ItemRef$ was obtained by the client, it provides the new address and the location of the new server where the item resides. 
In many cases, we only need one of the return values of $\GetItem$. Hence, for simplicity of presentation, we overload $\GetItem$ to return just $\Item$ or $\ItemRef$ alone. 

\item $\ItemRef \ \Lookup(key)$: The lookup operation takes a key and looks for an element with that key in the list. 
If a key is present in multiple items, all of the instances encountered during traversal are returned.
A $\Lookup$ may return `Not found' if no item with the given key has persisted for the entire duration of the $\Lookup$ operation. 


\item $\ItemRef \ \InsertAfter\ (\ItemRef, \Key)$: $\InsertAfter$ will insert a new element of the provided key at the referenced item and return the newly inserted element.

\item $\ItemRef \ \Next(\ItemRef)$: This will return the reference to the next element in the list. For example, if $\Next$(Ref(2)) is invoked on $\br{H, 5, 2, 7, 9, T}$, then the returned value will be Ref(7). Note that 
reference (as denoted by Ref, and later by $\ItemRef$)
is unique to a specific item in the list. 


\item $\Delete\ (\ItemRef)$: Logically delete the element pointed by $\ItemRef$. If the element was deleted earlier, the operation will be unsuccessful. $\Delete$ cannot be invoked on $\Head$ or $\Tail$. The deleted nodes are de-linked and garbage collected by a background server process (e.g., \cite{debra}, \cite{hazardPointers}), at a time that is most convenient for the server. It is only required that the garbage collection removes the node when there is no other thread referring to these deleted items. To save space, we illustrate the de-linking of deleted nodes in \Cref{sec:concurrentListAppendix}. 

\end{itemize}

\subsection{Context-Awareness}
Intuitively, a data structure is context-aware if it supports at least one operation which can only be invoked with a value returned by a previous operation on that data structure.

As an example, consider the implementation of an unordered linked list from \Cref{subsec:API}. In this example, the operation $\InsertAfter$ (respectively, $\Next$) can be invoked only if the client already has a reference to some element in the list (via a $\Lookup$ operation, $\Head$ operation, or previous $\InsertAfter$ or $\Next$ operation). Hence, the unordered linked list with the API from \Cref{subsec:API} is context-aware. If the syntax of $\InsertAfter$ was changed to leave this reference implicit ({as in the introduction, where the last client lookup provides the insertion context}), the list would still be a context-aware data structure.

By contrast, consider a simple hash table implementation with `\textit{set}' operations - Insert, Delete and Contains, where insert(key) (respectively, delete) inserts (respectively, deletes) the key to the hashtable and contains checks for a key in the hashtable. These operations do not require a value returned by a previous operation, and thus a simple hashtable is not context-aware. Similarly, a simple key value store with GET and PUT operations on its keys is not context-aware.

{We provide more examples and a mathematical definition of context-awareness in \Cref{sec:additionalcontextaware}.}

\subsection{Implementation of the Shared Memory Linked List}
\label{sec:baseimplementation}
\subsubsection{Algorithm}\label{subsec:ConcurrentListAlgorithm}

\Cref{alg:BaseAlgUpdate} showcases the implementation of the list operations. While the algorithm has been inspired from \cite{harris2001pragmatic}, the linked list in this paper is unordered.
The approach has also been simplified by having the mark bit in \cite{harris2001pragmatic} as the \textit{isDeleted} member for an item and insertions instead done by a Restricted Double Compare Single Swap (RDCSS)\cite{MultiCAS} operation. 
RDCSS takes five arguments: first two being a memory location and a value to compare its value against, and the last three being a memory location, a compared value and  the new value for the second location if both of the comparisons returned true. We interpret the returned value as true if both comparisons were true (hence value was updated), and false otherwise.

The $\InsertAfter$ operation first creates a temporary element and sets its \textit{next} item (Line \ref{alg:BaseAlgoUpdateInsertTemp}).
Then, $RDCSS$ is used to insert this item in the list. 
Specifically, the $RDCSS$ statement on Line \ref{alg:BaseAlgoInsertCAS} checks that $prev.next=temp$ (i.e., $prev.next$ has not changed between Lines \ref{alg:BaseAlgoUpdateInsertTemp} and \ref{alg:BaseAlgoUpdateInsertLoopContentEnd}), and that the item has not been deleted since the check on Line \ref{alg:BaseAlgoUpdateInsertLoopContentStart}). If this is satisfied, it sets $prev.next$ to be $temp$, thereby inserting the element in the list. 
If multiple elements are being inserted simultaneously, or the element gets deleted in the middle, {$RDCSS$ will fail}
and the $\InsertAfter$ operation will be retried. 
The $\Delete$ operation sets \textit{isDeleted} to true if the element was not deleted previously. If the element was deleted previously, it returns a failed status (Line \ref{alg:BaseAlgoDeleteFail}).
$\Lookup$ operations goes through the list to see if the key is present. Either it returns one or more $\ItemRef$(s) that contain the given key, or it returns `Not found'.
Finally, the implementation of $\Next$ identifies the \textit{next} element that is not yet deleted.

\begin{algorithm}
\scriptsize
\DontPrintSemicolon
\caption{A lock-free unordered linked list}\label{alg:BaseAlgUpdate}
\iffullalg
\Struct{\ItemRef} {
int id \Comment*[r]{serverId in which the node is located }
Address address \Comment*[r]{Physical Address of the node inside the server}
}
\else
\fi
\Struct{$\tt{Item}$} {
$\Key$ key,\;
$\ItemRef$ next,\;
$\tt{Boolean}$ isDeleted\;
}
\Fn{$\ItemRef$ \ $\InsertAfter$($\ItemRef$ prev, $\Key$ key)}{
    \Repeat{$\tt{rdcss}$(\&prev.isDeleted, false, \&prev.next, temp, newItem)\label{alg:BaseAlgoInsertCAS}}{
    \If{prev.isDeleted\label{alg:BaseAlgoUpdateInsertLoopContentStart}} {
        return ``Failed : Node not found"\;
    }
    
    $temp \gets prev.next$\;\label{alg:BaseAlgoUpdateInsertTemp}
    $newitem \gets new\ Item()$\;
    $newitem.next \gets temp$\;
    $newitem.key \gets key$\;\label{alg:BaseAlgoUpdateInsertLoopContentEnd}}
    \label{alg:BaseAlgoUpdateInsertAfterLoop}
    return newItem\;
}

\Fn{$\tt{Status}$ $\Delete$($\ItemRef$ prev)}{
    \Repeat{cas(\&prev.isDeleted, false, true)\label{alg:BaseAlgoDeleteCAS}} {
        \If{prev.isDeleted} {
            return ``Failed : Node not found"\;\label{alg:BaseAlgoDeleteFail}
        }
    }
    return ``Success"\;
}
\Fn{$\ItemRef$ \ $\Lookup$($\Key$ key)}{
    $resultSet \gets nil$\;
    $curr \gets \Head().next$\;
    \While{$curr.key \neq \tt{Tail}$}{
    \If{curr.key = key and curr.isDeleted = false }{
        resultSet.add(curr)\;
    }
    $curr \gets curr.next$\;
    }
    return resultSet\;
}

\Fn{$\ItemRef$ \ $\Next$($\ItemRef$ prev)} {
    $curr \gets prev$\;
    \Repeat{$curr.isDeleted = false$\label{alg:BaseAlgoNextIsDeleted}} {
        $curr \gets curr.next$\;
    }
    return curr\;
}
\end{algorithm}

\subsubsection{Properties of \Cref{alg:BaseAlgUpdate}}\label{subsec:LMLinearizability}

The properties of \Cref{alg:BaseAlgUpdate} are as follows. The operations $\Next$ (in a sublist), $\InsertAfter$ and $\Delete$ are \textit{linearizable}(\cite{linearizability}). 
These properties follow from the semantics of the \textit{CAS} and \textit{RDCSS} operations as explained in \Cref{sec:concurrentListAppendix}.

The property of $\Lookup$ is as follows: If an instance of key $k$ is present throughout the $\Lookup$ operation, then that instance of $k$ will be included in the response of $\Lookup$($k$). A $\Lookup$ may return `Not found' if no item with the given key has persisted for the entire duration of the $\Lookup$. (A linearizable $\Lookup$ requires coordination between all machines holding the list, and it is not desirable in a distributed data structure. Hence, we consider a weaker version of $\Lookup$.)

\section{Framework for Distribution}\label{sec:FrameworkForDistribution}
    In this section, we transform
\Cref{alg:BaseAlgUpdate} so that the list elements can be distributed across several machines.
We consider 
the model
where the list items are maintained by a set of servers.
While the number of servers can vary dynamically, we assume that the clients are aware of the identities of all the servers that are assigned to maintain the list. Clients request operations, and perform them in one or more servers. We assume that the pool of servers, and their network locations are known to the client. In practice, this information can be hidden from the clients by abstracting  it to a middle-ware service. We require the channels to be reliable, but they do not need to deliver messages in any particular order.

\subsection{Distributed Client-Server Architecture}\label{subsec:approachForDistbn}
In our work, a list is divided into multiple \sublists, and these \sublists are distributed across a set  of servers. 
Each \sublist has dummy elements $\subhead$ and $\subtail$.
Each server maintains a $\Registry$ that provides the $\subheads$ of all the \sublists that it \textit{owns}.


\begin{figure}
    \centering
    \includegraphics[scale=0.5]{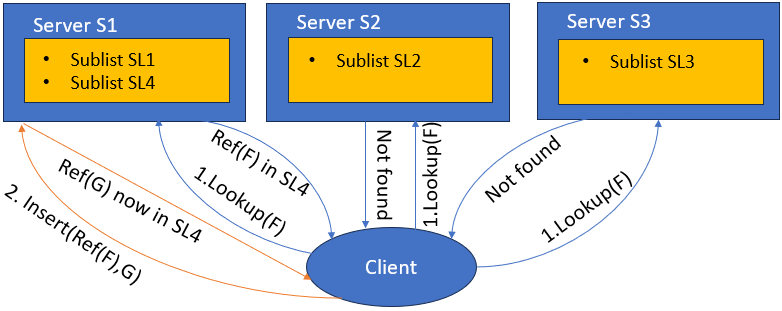}
    \caption{$\Lookup$ (blue) for items from all active servers and $\InsertAfter$ (red) requested at the server that owns the item.}
    \label{fig:architecture}%
\end{figure}

\begin{figure}
    \centering
    \includegraphics[scale=0.4]{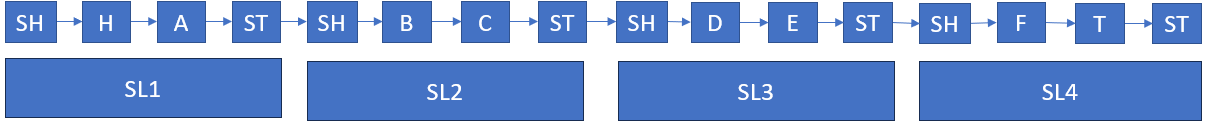}
    \caption{A list being composed of a set of 4 \sublists.}
    \label{fig:archSublists}
    \vspace{-3mm}
\end{figure}

Consider a linked list that is composed of 4 different \sublists as shown in \Cref{fig:archSublists}. 
Here, the list is $\br{A, B, C, D, E, F}$ that is split into four \sublists, $SL_1$, $SL_2$, $SL_3$ and $SL_4$. The \sublists are on servers $S_1$, $S_2$ and $S_3$ as shown in \Cref{fig:architecture}. 

To define the list API, we rely on an opaque data type $\ItemRef$ (or Ref, for short). Specifically, $\ItemRef$ for an item provides a unique reference to that item. One possible implementation idea is that $\ItemRef$ can be a combination of the physical address and the id of the server it resides in.

As illustrated in \Cref{fig:architecture}, at the time of a $\Lookup$ for a key $F$, a client broadcasts its request to all servers. The client will receive matching nodes from every server that has a node of the requested key($F$) and a `\textit{Not Found}' type response from those that do not have the key. A stub added to the client will read all these responses and decide on a node of its choice, possibly from reading payload data contained within the node. 

An $\InsertAfter$ operation, on the other hand, uses $\ItemRef$ returned from a previous operation. For example, if we do a lookup for key $F$ then it will return a reference to an element in $SL_4$. $\InsertAfter$(Ref(F), G) will now only be sent to server $S_1$ that is hosting $SL_4$.


If \sublists maintained at different servers grow/shrink at different sizes then it would be necessary to rebalance the load across different servers. This will require additional client-invisible mechanisms. We discuss them in \Cref{sec:Requirements}. 


\subsection{Requirements of the Transformation:}
\label{sec:Requirements}
\begin{figure}
    \centering
    \includegraphics[scale=0.5]{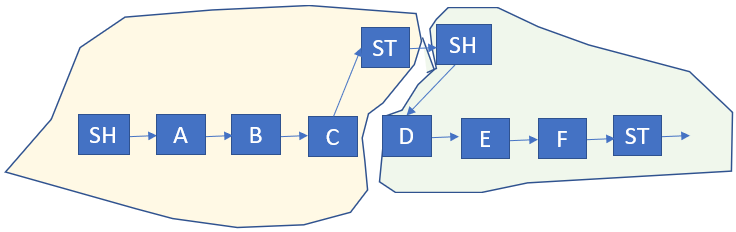}
    \caption{Splitting of a \sublist inserts a block containing two nodes(ST and SH) at the chosen point of \Split (C)}
    \label{fig:splitsublist}
    \vspace{-3mm}
\end{figure}
{When a distributed list contains a \sublist $SL$, that grew beyond a certain size, it needs to be \textit{split} into two \sublists (say, $SL_1$ and $SL_2$). Additionally, when the server hits its throughput limit, it can \textit{move} some of its  \sublists to a different server, to increase the throughput of the distributed list.}

With this motivation, we introduce three new client-invisible operations, $\Split$, $\Move$ and $\SwitchLB$. The semantics are:
    
    
    \begin{itemize}[nosep, leftmargin=*]
        \item \textbf{$(\ItemRef,\ItemRef )\ \Split\ (\subhead, \tt{node})$:}
        The requirement of the $\Split$ operation is to successfully split a given \sublist into two new smaller \sublists. This is achieved by inserting a $\subtail$ followed by a $\subhead$ at the node where the given \sublist is to be separated. This new $\subhead$ is added to the $\Registry$ and the inserted dummy nodes are then returned.
        For example, if $\Split$ is called on the list in Figure \ref{fig:splitsublist} on item \textit{C}, then it will result in two \sublists, $SL_1=\br{A, B, C}$ and $SL_2 =\br{D, E, F}$. 

        \item \textbf{$\Move\ ({{\subhead}, \tt{serverId}})$:} The requirement of the $\Move$ operation is to successfully move all nodes of a \sublist from a server, say $S_1$, to a server given by the id of the server, say $S_2$. At the end of the $\Move$ operation, the \sublist present in $S_2$ should be the most recent version of the \sublist and any future update on the \sublist should be reflected in $S_2$.        
       
        \item \textbf{$\SwitchLB\ ({\tt{\subhead}, \tt{serverId}, \tt{\subhead}})$:}  The $\SwitchLB$ operation takes in the local and remote subheads of the moved \sublist, seamlessly de-link the local (stale) \sublist, and replace it with the moved (active) \sublist without losing availability of the \sublist during the process. After the switch, the stale \sublist must be free to be physically deleted.
        
    \end{itemize}
    Additionally, we require the client-visible operations to remain lock-free at all times.

\noindent{\bf Assumptions: }\label{para:Assumptions}
\begin{itemize}[nosep, leftmargin=*]
    \item There should be hardware support for single-word CAS instructions. This inturn helps us utilize RDCSS and AtomicInteger as linearizable CAS constructions. 
    \item 
We do not make assumptions on how a server decides when a \sublist should be $\Split$ or which \sublist should be moved and where. Some possible approaches are (1) keeping a count to limit the number of items in a \sublist, 
(2) ensuring that the number of items on different servers are \textit{close to one another}. We only focus on ensuring that these client-invisible operations do not violate the properties of the underlying data structure.

\item
Since client-invisible operations are self-invoked by the server, we assume that at most one operation from $\Split$, $\Move$, and $\SwitchLB$, is active on a given \sublist.
Additionally, with the same reasoning, we disable any physical deletion mechanisms during these operations. 

\item 
We require each $\ItemRef$ returned to have a maximum expiration time and a request timeout period. We denote the sum of these two quantities by $\theta$. If a client needs to keep the referenced item
for a longer duration, it needs to \textit{refresh} it using $\GetItem$. {We expect this to be easily achieved by a client stub.}


\end{itemize}
    
    


\subsection{Key Challenges}
    In this section we go over the difficulty in designing the transformation operations of \Cref{sec:Requirements}. The main difficulty comes from designing the $\Move$ and $\SwitchLB$ operations in a way that the client-visible operations still continue to make \textbf{lock-free progress} during those transformation operations. A $\Move$ operation cannot just use a snapshotting procedure on a \sublist using approaches such as  \cite{lockfreeIterators}, as the snapshot thus created should continue to gain the effect of all writes ($\InsertAfter$, $\Delete$) that were invoked on the \sublist during, and after the $\Move$.
    
    Another difficulty comes from the context-awareness of the data structure. The distribution and any movement of data must take into account the client states, i.e., the state of values returned by previous operations that the client may still be holding. These client operations must work correctly and in a lock-free manner even while load balancing the data structure. Inspite of this, the $\SwitchLB$ operation needs to retain the guarantees of client operations and preserve the structure of the linked list. A non context-aware data structure, such as a hashtable, does not face this challenge. 






\section{Aborting Move (\AbortingMoveMethod) protocol}
\label{sec:AMSection}

In this section, we present our \AbortingMoveMethod protocol for a distributed list that does not use locks. Rather, it relies on prioritizing client-visible operations over client-invisible operations. 


\subsection{General Idea}\label{subsec:AMGeneralIdea}
    The general idea behind \AbortingMoveMethod protocol is to treat the $\Move$ operation as the operation of least priority, that could be aborted in case of a concurrent update.
    To achieve this, we use atomic counters, \StartCount and \EndCount, 
    for each \sublist. 
    We initialize these counters to 0. All update operations ($\InsertAfter$, $\Delete$) atomically increment the \StartCount (respectively, \EndCount) when they start (respectively, finish).


    A $\Move$ operation on a \sublist from $S_1$ to $S_2$, first, makes a copy of \EndCount. 
    At the end of the move, this value of \EndCount and the current value of \StartCount are compared to determine if an update operation occurred during the move. 
    If an update operation occurred on any element in the \sublist, the $\Move$ operation is restarted. 
    On the other hand, if no update operations occurred, then all future
    updates are sent to $S_2$.

\subsection{Data Structure of an Item}\label{subsec:AMDS}

In addition to \Cref{alg:BaseAlgUpdate}, each node has two pointers to the $\AtomicInteger$ variables, \StartCount and \EndCount,  associated with the \sublist the node is present in. These counts, the $\subtail$ of the previous list, and another integer, \textit{offset}, will reside in the $\Registry$ as properties of the \sublist corresponding to that entry.
Additionally, a $newLocation$ variable is present in each node to provide the new location for the item in case it has been moved to another server. It is of type $\ItemRef$, and so has a \textit{serverID} and a physical address contained within it.

\iffullalg
\begin{algorithm}
\scriptsize
\DontPrintSemicolon

\caption{\AbortingMoveMethod Protocol Data Structure}\label{alg:AMAlgoStruct}
\Struct{Item} {
\iffullalg
\Key key\;
\ItemRef next\;
\else 
Include key, isDeleted and next from \Cref{alg:BaseAlgUpdate}\;
\fi
$\AtomicInteger$ *\StartCount,  *\EndCount\; \Comment*[r]{ Here * indicates reference-type members}
$\ItemRef$ \ \newLocation \Comment*[r]{Remote node location}
}
$\Registry$: \textit{Each entry refers to a \sublist by storing its $\subhead$ pointer as key and a tuple value containing \{\StartCount, \EndCount, \offset , the subtail of its previous \sublist\}}\;\label{alg:AMAlgoStructRegistry} 
\end{algorithm}
\fi
\subsection{Algorithm}\label{subsec:AMAlgorithm}
In this section, we describe the \AbortingMoveMethod protocol in detail. The pseudo code for the methods in this protocol are in Algorithms \ref{alg:AMAlgoLoadBalanceSplit}, \ref{alg:AMAlgoUpdate}, \ref{alg:AMAlgoRead} and \ref{alg:AMAlgoLoadBalanceMove}. We describe these operations, as follows: 



\noindent 
\textbf{$\Split$ (\Cref{alg:AMAlgoLoadBalanceSplit}):} 
The $\Split$ operation  creates a $\dummyNode$ and a $\subhead$ that is next to the $\dummyNode$. It then inserts the $\dummyNode$ into the \sublist (Line \ref{alg:AMAlgoLoadBalanceSplitCAS}).

Before the new \sublist is added to the $\Registry$, it additionally needs its own \StartCount and \EndCount variables. Hence, two new counter variables (\newStartCount and \newEndCount) are created (Lines \ref{alg:AMAlgoLoadBalanceSplitStartCount} and \ref{alg:AMAlgoLoadBalanceSplitEndCount}). After the $\Split$ of a \sublist $SL$ into $SL_1$ and $SL_2$, \sublist $SL_1$ will use the counters (\StartCount and \EndCount) associated with the original \sublist $SL$ and the \sublist $SL_2$ will use the new counters. Hence, pointers to all elements in $SL_2$ are changed to \newStartCount and \newEndCount.

        
        During the course of this operation, it is possible that an update operation increments \StartCount of $SL$ but increments \newEndCount that will belong to $SL_2$. This violates the expected invariant that $\StartCount - \EndCount$ of a \sublist is $0$ when no update operation is pending. 

        We fix the problem as follows: instead of maintaining the invariant ($\StartCount - \EndCount = 0$) when no update operation is pending, we maintain ($\StartCount - \EndCount =\offset$), where \offset is a parameter of the \sublist that only changes when the \sublist is $\Split$ . Given the original \sublist $SL$, we have $\StartCount-\EndCount=\offset_{SL}$ when no update operation is pending. Even if an update operation increments \StartCount but increments \newEndCount, the following invariant will remain true during the $\Split$ operation if no update operation is pending:
        $(\StartCount-\EndCount) + (\newStartCount-\newEndCount) =\offset_{SL}$. 
        Furthermore, $(\StartCount-\EndCount) + (\newStartCount-\newEndCount) > \offset_{SL}$ if an update is pending. 
        
        Hence, we compute $(\StartCount-\EndCount) + (\newStartCount-\newEndCount)$ to determine if it is equal to $offset_{SL}$. When they are equal, we know that no update operation is pending on $SL_1$ and $SL_2$, thereby finding the new \offsets on a snapshot. Thus, \offsets of $SL_1$ and $SL_2$ can be set to the computed $(\StartCount-\EndCount)$ and $(\newStartCount-\newEndCount)$ respectively. 
        
        Finally, $\Split$ adds the new $\subhead$ to the $\Registry$ and truncates the first half of the original \sublist by re-labelling the $\dummyNode$ as a $\subtail$. The node was initially a $\dummyNode$ in order to not truncate $\Lookup$ operations until the next \sublist is added into the $\Registry$ in Line \ref{alg:AMAlgoLoadBalanceSplitAddReg}.
        

\noindent \textbf{$\InsertAfter, \Delete$ (\Cref{alg:AMAlgoUpdate}): } 
$\InsertAfter$ and $\Delete$ increment the \StartCount associated with the \sublist before executing, and increments \EndCount after executing the operation. If the value of the \StartCount is negative, it indicates that the node has been moved (the underlying rationale is explained in $\Move$). Hence, the operation is forwarded or \textit{delegated} to the new machine where it is now active. Otherwise, the operation is identical to any lock-free implementation of the corresponding operation.
        
\noindent   \textbf{$\Lookup, \Next$ (\Cref{alg:AMAlgoRead}): } $\Lookup$ and $\Next$ operations additionally check if \StartCount is negative, to know if the list has been moved. 
        Subsequently, $\Lookup$ traverses each sublist until it reaches $\subtail$.
        If an matching key is found and is not \textit{stale} (i.e., the \StartCount is not negative) then it is added to the set of items that are returned by \Lookup. 
        If stale, the result of $\Lookup$ is fetched from the machine where the item is currently active.       

        In the case of $\Next$ operation, a distributed list traversal additionally requires handling of next pointer when encountering a $\subtail$ as the immediate next node. When not nil, this next pointer points to the $\subhead$ of the next \sublist, and so the operation delegates the request as $\Next$(\textit{$\subhead$}) to the server given by the server ID parameter of the next pointer in the $\subtail$.

\noindent \textbf{$\Move$  (\Cref{alg:AMAlgoLoadBalanceMove}): } $\Move$ operation moves a \sublist from a server $S_1$  to a server $S_2$. 
First, the current \EndCount value of the \sublist is stored in $\counterTemp$ (Line \ref{alg:AMAlgoLoadBalanceMoveCounterTemp}), before attempting to move the entire \sublist one node at a time.   
By the property of \offset discussed above, when all elements in the \sublist are copied, and $\StartCount- \counterTemp = \offset$, it indicates that no update operations overlapped with the $\Move$ operation. In such a scenario, the copy of the list on $S_1$ and $S_2$ are identical at Line \ref{alg:AMAlgoLoadBalanceMoveCAS}. Hence, $\Move$ completes successfully and sets \StartCount to $-\infty$, thereby forcing \StartCount to stay negative after any future increments. From this point onwards, client-visible operations for the \sublist will observe \StartCount to be negative and will delegate their requests to $S_2$.
If the \textit{CAS} in Line \ref{alg:AMAlgoLoadBalanceMoveCAS} fails, $\Move$ operation is aborted (by deleting the moved nodes from $S_2$) and restarted (by Lines \ref{alg:AMAlgoLoadBalanceMoveSLDelete} and \ref{alg:AMAlgoLoadBalanceMoveRetry}, respectively) .
        
\noindent \textbf{$\SwitchLB$ (\Cref{alg:AMAlgoLoadBalanceMove}): } $\SwitchLB$ operation switches the ownership of a \sublist (say $SL_2$) to a moved server (say $S_2$). To achieve this, $S_1$ de-links its \sublist from the list traversal, by making the $\subtail$ of its preceding \sublist point to the moved version in $S_2$ (Line \ref{alg:AMAlgoLoadBalance:switchPrecedingSubTail}), and requests $S_2$ to add $SL_2$ to its $\Registry$, thereby enabling $S_2$ to serve $\Lookup$ operations of $SL_2$ (Line \ref{alg:AMAlgoLoadBalance:switchRequest}).
After this, $S_1$ has to only wait until all clients lose references to its stale nodes (Line \ref{alg:AMAlgoLoadBalance:switchWait}) and then safely remove $SL_2$ from its $\Registry$ and physical memory.

\noindent \textbf{$\tt{Delegations}$: } 
 Delegations are messages sent when switching ownership of a \sublist from $S_1$ to $S_2$. 
During $\SwitchLB$ operation, $S_2$ is the server containing the active references of the \sublist. 
However, requests may arrive at $S_1$, as clients do not know about the $\Move$. Delegations forward these requests to $S_2$.
For $\Next$, $\InsertAfter$ and $\Delete$ operations, the delegation messages have client information attached, so $S_2$ can perform the corresponding operation and return the response directly to the client. In the case of $\Lookup$, the matching node for the requested \sublist is sent back to $S_1$, which in turn returns it to the client as part of its set of matching nodes. This maintains the $\Lookup$ property that each server sends one response to the client, containing a set of key-matching nodes.
\begin{algorithm}
\scriptsize
\DontPrintSemicolon
\caption{$\Split$ in \AbortingMoveMethod Protocol}\label{alg:AMAlgoLoadBalanceSplit}
\Fn{($\ItemRef$,$\ItemRef$) \textbf{$\Split$}($\ItemRef$ SL, $\ItemRef$ node)} {
    $ST \gets new\ \dummyNode()$\;
    $SH \gets new\ \subhead()$\;
    $ST.next \gets SH$\;
    \Repeat{cas(\&node.next, next, ST)\label{alg:AMAlgoLoadBalanceSplitCAS}} { 
        $next \gets node.next$\;
        $SH.next \gets next$\;
    }
    $curr \gets SH$\;
    $\newStartCount \gets new\ \AtomicInteger(0)$\;
    $\newEndCount \gets new\ \AtomicInteger(0)$\;
    \Repeat{$prev.key = \subtail$} {
        $prev \gets curr$\;
        $curr.\StartCount \gets \&\newStartCount$\; \label{alg:AMAlgoLoadBalanceSplitStartCount}
        $curr.\EndCount \gets \&\newEndCount$\; \label{alg:AMAlgoLoadBalanceSplitEndCount}
        $curr \gets curr.next$\;
    }
    \label{alg:AMAlgoLoadBalanceSplitCounterUpdateFinished}
    \Repeat{$a_1 + a_2 = \Registry.get(SL).getOffset()$ \label{alg:AMAlgoLoadBalanceSplitOffsetCompute}} {
        $a_1 \gets \newStartCount - \newEndCount$\;
        $a_2 \gets node. \StartCount \PArrow get() - node. \EndCount \PArrow get()$\;
    }
    $\Registry$.add(SH, \{\newStartCount, \newEndCount, $a_1$, ST\})\; \label{alg:AMAlgoLoadBalanceSplitAddReg}
    $\Registry$.get(SL).setOffset($a_2$)\; \label{alg:AMAlgoLoadBalanceSplitUpdateReg}
    $ST.key \gets \subtail$\; \label{alg:AMAlgoLoadBalanceSplitTruncate}
    return $(ST,SH)$\;
}

\end{algorithm}
\begin{algorithm}
\scriptsize
\DontPrintSemicolon
\caption{Updates in \AbortingMoveMethod Protocol}\label{alg:AMAlgoUpdate}
\Fn{$\ItemRef$ \textbf{$\InsertAfter$}($\ItemRef$ prev, $\Key$ key)}{
    $prev.\StartCount\PArrow increment()$\; \label{alg:AMAlgoUpdateInsertStartCount}
    \If{$prev.\StartCount\PArrow get() < 0$\label{alg:AMAlgoUpdateInsertCheck}} {
        Send DelegateInsertAfter(prev.$\newLocation$, key, clientInfo) to prev.$\newLocation$.getSId()\;\label{alg:AMAlgoUpdateInsertDelegate}
        return \;
    }
    \Repeat{$\tt{rdcss}$(\&prev.isDeleted, false, \&prev.next, temp, newItem)\label{alg:AMAlgoUpdateInsertCAS}}{
    \iffullalg
    \If{prev.isDeleted} {
        return ``Failed : Node not found"\;
    }
    
    $temp \gets prev.next$\;
    $newitem.next \gets temp$\;
    $newitem.newLocation \gets prev.newLocation$\;
    $newitem.key \gets key$\;
    $newItem.isDeleted \gets false$\;
    $newItem.\StartCount \gets prev.\StartCount $\;
    $newItem.\EndCount \gets prev.\EndCount $\;
    \else
    Create a $newitem$ of given $key$ and copy all other item values from $prev$, similar to Lines \ref{alg:BaseAlgoUpdateInsertLoopContentStart} to \ref{alg:BaseAlgoUpdateInsertLoopContentEnd} of \Cref{alg:BaseAlgUpdate}
    \fi
    }
    \label{alg:AMInsertBeforeEndCount}
    $prev.\EndCount\PArrow increment()$\;\label{alg:AMAlgoUpdateInsertEndCount}
    return newItem\;
}

\Fn{\textbf{$\Delete$}($\ItemRef$ prev)}{
    $prev.\StartCount\PArrow increment()$\;\label{alg:AMAlgoUpdateDeleteStartCount}
    \If{$prev.\StartCount\PArrow get() < 0$ \label{alg:AMAlgoUpdateDeleteCheck}} {
        Send DelegateDelete(prev.$\newLocation$, clientInfo) to prev.$\newLocation$.getSId() \label{alg:AMAlgoUpdateDeleteDelegate}\;
        return \;
    }
    \Repeat{cas(\&prev.isDeleted, false, true)\label{alg:AMAlgoUpdateDeleteCAS}} {
        \If{prev.isDeleted} {
        return ``Failed : Node not found"\;
        }
    }
    \label{alg:AMDeleteBeforeEndCount}
    $prev.\EndCount\PArrow increment()$\; \label{alg:AMAlgoUpdateDeleteEndCount}
    return ``Success"\;
}

\end{algorithm}
\begin{algorithm}
\scriptsize
\DontPrintSemicolon

\caption{$\Lookup$, $\Next$ in \AbortingMoveMethod Protocol}\label{alg:AMAlgoRead}
\Fn{$\ItemRef$[] $\Lookup$($\Key$ key)}{
    $result \gets nil$\;
    \For{(SH in $\Registry$.getSubheads())} {
        $result \gets result \cup (LookupSublist(SH, key))$\;
    }
    return result\;
}
\Fn{$\ItemRef[]$\ $\tt{LookupSublist}$($\ItemRef$ head, $\Key$ key)}{
    $resultSet \gets nil$\;
    \If{$curr.\StartCount\PArrow get() < 0$\label{alg:AMAlgoReadLookupDelegateCheck1}}{ 
        return SendDelegateLookup(head.$\newLocation$, key) to head.$\newLocation$.getSId()\; \label{alg:AMAlgoReadLookupDelegate1}
    }
    $curr \gets head.next$\;
    \While{$curr.key \neq \subtail$}{
    \If{curr.key = key and  curr.isDeleted = false and $curr.\StartCount\PArrow get()$ $\ge$ 0 \label{alg:AMAlgoReadLookupValidityCheck}}{
        $resultSet.add(curr)$\;
    }
    $curr \gets curr.next$\;
    }
    \If{$curr.\StartCount\PArrow get() < 0$ \label{alg:AMAlgoReadLookupDelegateCheck2}}{
        return SendDelegateLookup(head.$\newLocation$, key) to head.$\newLocation$.getSId()\; \label{alg:AMAlgoReadLookupDelegate2}
    }
    return resultSet\;
}

\Fn{$\ItemRef$ \textbf{$\Next$}($\ItemRef$ prev)} {
    $curr \gets prev$\;
    \Repeat{$curr.key \notin \{\dummyNode, \subhead\}$ \label{alg:AMAlgoReadNextElement} and $curr.isDeleted = false$} {
        $curr \gets curr.next$\; 
    }
    \If{$curr.\StartCount\PArrow get() < 0$} { \label{alg:AMAlgoReadNextDelegateCheck}
        Send DelegateNext(prev.$\newLocation$, clientInfo) to prev.$\newLocation$.id\; \label{alg:AMAlgoReadNextDelegateByMove}
        return\;
    }
    \If{$curr.key = \subtail$} {
    \If{$curr.next.getSId() = myId$} {
        return $\Next$(curr.next)\; \label{alg:AMAlgoReadNextSH}  
    }
    Send DelegateNext(curr.next, client) to curr.next.getSId()\;
    \label{alg:AMAlgoReadNextDelegateByST}
    return \;
    }
    return curr\;
}
\end{algorithm}
\begin{algorithm}
\scriptsize
\DontPrintSemicolon
\caption{$\Move$ in \AbortingMoveMethod Protocol }\label{alg:AMAlgoLoadBalanceMove}

\Fn {$\Move$($\ItemRef$ head, $\tt{Integer}$ newServerId)} {
    $\counterTemp \gets head.\EndCount\PArrow get() + \Registry.get(head).getOffset()$\; \label{alg:AMAlgoLoadBalanceMoveCounterTemp}
    $remoteSH \gets \tt{MoveItem}(nil, head, newServerId)$\;
    $curr \gets head$\;
    $remoteRef \gets remoteSH$\;
    \Repeat{$curr.key = \subtail$}{
        $curr \gets curr.next$\;
        $remoteRef \gets \tt{MoveItem}(remoteRef, curr, newServerId)$\;
    }
    \eIf{cas(\&head.\StartCount, \counterTemp, $-\infty$)\label{alg:AMAlgoLoadBalanceMoveCAS}} {
        $\SwitchLB$(head, newServerId, remoteSH)\;
    } {
        Instruct the server with server ID as newServerId, to delete the moved nodes starting from remoteSH, and wait for its acknowledgement\;\label{alg:AMAlgoLoadBalanceMoveSLDelete}
        $\Move$(head, newServerId)\;\label{alg:AMAlgoLoadBalanceMoveRetry}
    }
}

\Fn{$\ItemRef$ $\tt{MoveItem}$($\ItemRef$ remotePrevRef, $\ItemRef$ currRef, $\tt{Integer}$ newServerId)} {
$response \gets Send Move(remotePrevRef, currRef.getItem())$ to newServerId\;
$currRef.\newLocation \gets response$\;
return response\;
}

\Fn{\textbf{$\SwitchLB$}($\ItemRef$ head, $\tt{Integer}$ newServerId, $\ItemRef$ remoteSH) \label{alg:AMAlgoLoadBalance:switchStart}} {
    $ST \gets \Registry.get(head).getPrevSubtail()$\;
    \eIf{$ST.\StartCount\PArrow get() < 0$} {
        Request responsible server from ST.\newLocation.getSId() to update next of ST.newLocation to remoteSH and wait for acknowledgement\;
    } {
    $ST.next \gets remoteSH$\; \label{alg:AMAlgoLoadBalance:switchPrecedingSubTail}
    }
    $response \gets $Send $\SwitchLB$(head.\newLocation, ST) to newServerId\; \label{alg:AMAlgoLoadBalance:switchRequest}
    Wait for $\theta$ time\; \label{alg:AMAlgoLoadBalance:switchWait}
    $\Registry$.remove(head)\; \label{alg:AMAlgoLoadBalance:switchRemoveReg}
    \iffullalg
    $curr \gets head$ \Comment*[r]{Mark \sublist for deletion}\label{alg:AMAlgoLoadBalance:switchMarkForDeletion}
    \Repeat{$curr = \subtail$}{
        $curr.isDeleted \gets true$\;
        $curr \gets curr.next$\;
    } \label{alg:AMAlgoLoadBalance:switchEnd}
    \else
    Mark \sublist for deletion \;\label{alg:AMAlgoLoadBalance:switchMarkForDeletion}
    \fi
}

\Fn{$\ItemRef$ $\tt{MoveReceive}$($\ItemRef$ prev, $\tt{Item}$ item)}{
    \eIf{item.key = $\subhead$}{
        $SH \gets new\ \subhead()$\; \label{AMMoveReceiveSHStart}
        $ST \gets new\ \subtail()$\;
        $\newStartCount \gets new\ \AtomicInteger(0)$\;
        $\newEndCount \gets new\ \AtomicInteger(0)$\;
        $SH.\StartCount, ST.\StartCount \gets \&\newStartCount$\;
        $SH.\EndCount, ST.\EndCount \gets \&\newEndCount$\;
        $SH.next \gets ST$\; 
        return SH\;\label{AMMoveReceiveSHEnd}
    }{
        \eIf{item.key = $\subtail$}{
            $(prev.next).next \gets item.next$\;
            return prev.next\;
        }{
            $newItem \gets \InsertAfter(prev, item.key)$\;
            return newItem\;
        }
    }
}

\Fn{$\tt{SwitchReceive}$($\ItemRef$ head, $\ItemRef$ prevSubtail)} {
$\Registry$.add(head, \{head.\StartCount, head.\EndCount, 0, prevSubTail\})\;
\label{alg:AMAlgoLoadBalanceReceiveSwitchLine}
return ``Success"\;
}
\end{algorithm}
\section{Temporary Replication (\TempReplMethod) Protocol}\label{sec:TRSection}

In \TempReplMethod protocol,
instead of aborting the $\Move$ operation as in the \AbortingMoveMethod protocol, we let the update operations \textit{help} complete the $\Move$ in a way that does not impact the latency of the update request, as perceived by the client.

\subsection{General Idea}\label{subsec:TRGeneralIdea}
Consider the case where a \sublist $SL$, as described in \Cref{fig:sampleList} is being moved from server $S_1$ to server $S_2$.
Similar to previous methods, \TempReplMethod  moves the \sublist in order, i.e., first, the element $A$, then $B$ and so on. If there are updates done at an item that was moved, then that change is communicated to $S_2$ by a \textbf{$\Replicate$} message. 
Until the $\SwitchLB$ operation, all client requests ($\Lookup$, $\InsertAfter$, etc.) are completed by $S_1$. 
$\SwitchLB$ to $S_2$ occurs when all $\Replicate$ messages have been \textit{replayed} in $S_2$, at which point the \sublist $SL$ of $S_1$ and $S_2$ are momentarily identical.
(Detecting a point of no pending $\Replicate$ messages is discussed in \Cref{subsec:TRAlgorithm}.)
\TempReplMethod and \AbortingMoveMethod share the same $\Split$ and $\SwitchLB$ operations. However, the $\Move$ operation in \TempReplMethod becomes more involved to accommodate $\Replicate$ messages.

\textbf{Key Challenge in $\TempReplMethod$ protocol: }
Continuing with the example above, suppose  
the $\Move$ operation is moving element $B$ and a client invokes $\InsertAfter$(Ref(D), G) (respectively, $\Delete$(Ref(D))), then $G$ will be inserted (respectively, $D$ will be deleted) on $S_1$. The $\Move$ operation will move $D$ while moving the elements one at a time. On the other hand, if there is $\InsertAfter$(Ref(A), G) then $\Move$ operation will not move the new element $G$. Instead, it will be transferred by a \textit{$\Replicate$} message. It is possible that there are concurrent $\InsertAfter$(Ref(A), $N_1$) and $\InsertAfter$(Ref(A), $N_2$) calls. If this happens, the inserting order of $N_1$ and $N_2$ must be identical on $S_1$ and $S_2$, even if the corresponding $\Replicate$ messages reach $S_2$ in an arbitrary order. 

In the above example, $\InsertAfter$(Ref(A), $N_1$) and $\InsertAfter$(Ref(A), $N_2$) can be processed as soon as they are received. However, if there is an $\InsertAfter$(Ref($N_1$), $N_3$) then this must be processed after $\InsertAfter$(Ref(A), $N_1$). Likewise, if there is a $\Delete$(Ref($N_1$)) then it must be processed after $\InsertAfter$(Ref(A), $N_1$). 

    Note that the above challenge cannot be overcome with just FIFO channels, because a thread can be preempted just after the \textit{RDCSS} operation (that adds the element), which will delay the sending of the corresponding $\Replicate$ message.
\subsection{Data Structure of an Item}\label{subsec:TRDS}
    On top of the structure in \Cref{subsec:AMDS}, the \TempReplMethod requires a logical timestamp and a server ID (of inserting server) as additional members of its structure. The logical timestamp helps order insertion events that take place in a \sublist and the server ID, along with the logical timestamp, helps to uniquely identify a node in the remote server.

\iffullalg
\begin{algorithm}
\scriptsize
\DontPrintSemicolon
\caption{\TempReplMethod Protocol Data  Structure}\label{alg:TRAlgoStruct}
\Struct{Item} {
\iffullalg
\Key key\;
\ItemRef next\;
\AtomicInteger \&\StartCount\;
\AtomicInteger \&\EndCount\;
\else
Include key, isDeleted, next, \newLocation, \StartCount and \EndCount from \Cref{alg:AMAlgoStruct}\;
\fi
$\tt{Integer}$ ts, \Comment*[r]{logical timestamp}
$\tt{Integer}$ sId \Comment*[r]{ID of the inserting server}

\iffullalg
\ItemRef \newLocation \Comment*[r]{Remote node location information}
\else
\fi
}
$\Registry$: \textit{Same as Line \ref{alg:AMAlgoStructRegistry} in \Cref{alg:AMAlgoStruct}\;
Each server additionally has a logical clock \LC that gets incremented during insertions\;}
\end{algorithm}
\fi
\subsection{Algorithm}\label{subsec:TRAlgorithm}
The pseudo code of \TempReplMethod protocol shares the same pseudo code as \AbortingMoveMethod protocol for $\Split$, $\Lookup$, $\Next$ and $\SwitchLB$ operations. The pseudo code in Algorithms \ref{alg:TRAlgoUpdateOps}, \ref{alg:TRAlgoLoadBalanceMove} and \ref{alg:TRAlgoReplicateUpdateReceives} describe the following operations that remain:

\noindent  \textbf{$\InsertAfter, \Delete$ (\Cref{alg:TRAlgoUpdateOps}): } 
The update operations of \TempReplMethod protocol is the same as in the \AbortingMoveMethod protocol in the absence of a concurrent $\Move$. After returning a response to the client, in \TempReplMethod protocol, the updates
will also send $\Replicate$ messages (Lines \ref{alg:ReplicateMsgInsert} and \ref{alg:ReplicateMsgDelete}) of the update event when there is an ongoing $\Move$. During the $\Move$, the \EndCount increment is instead done after the update replay is successful (Lines \ref{alg:ReplicateInsertAfterEndC} and \ref{alg:ReplicateDeleteEndC}). To support replay of inserts in the right order, $\InsertAfter$ operation assigns a logical timestamp to the item being inserted (Line \ref{alg:TRAlgoUpdateTS}). Note that if \textit{CAS} fails, then the $\InsertAfter$ operation obtains a new timestamp. This ensures that if $\InsertAfter$(prev, $N_1$), $\InsertAfter$(prev, $N_2$) are invoked, and $N_1$ is inserted before $N_2$, then the timestamp of $N_1$ will be lower than that of $N_2$. These timestamps are used during $\Move$, to help $S_2$ insert the items in the same order as $S_1$. 

    \noindent \textbf{$\Move$ (\Cref{alg:TRAlgoLoadBalanceMove}): } {The $\Move$ operation at a server $S_1$ calls MoveItem at Line \ref{alg:TRAlgoMoveRecurse}, to recursively traverse and move the \sublist to a server $S_2$. It will return the remote $\subhead$ after $\Move$, to begin $\SwitchLB$ phase.}
    During the $\Move$ process, concurrent update operations that occur on the moved items of the \sublist, have to be replicated, before the $\Move$ phase is completed. In order to achieve this, $S_1$ sends $\Replicate$ messages (described later) to $S_2$ for each update operation. $S_2$ will use the update timestamp to replay the replicated event and reconstruct the \sublist present in $S_1$. {Thus, the $\Move$ operation formally completes when both the recursive traversal has reached the $\subtail$ of the \sublist and additionally, $S_2$ has communicated the successful replay of each $\Replicate$ message $S_1$ had sent to it.}
    When the \textit{CAS} of Line \ref{alg:TRAlgoLoadBalanceMoveCAS} first 
    detects this point, \StartCount is set to {-$\infty$} by that \textit{CAS} (as in the \AbortingMoveMethod protocol) to complete the $\Move$ (Line \ref{alg:TRMoveWaitforReplays}) and begin the $\SwitchLB$ operation (Line \ref{alg:TRSwitchCall}).
    
    \noindent \textbf{$\Replicate$ (\Cref{alg:TRAlgoReplicateUpdateReceives}): } $\Replicate$ is not a new operation, but a type of message sent upon doing operations at a node marked {with a $\newLocation$ instead.}
    {$\Replicate$ messages ensure that updates concurrent with a $\Move$ do not go unnoticed by $S_2$.}
    $S_2$ then picks up on these updates and replays the events to reconstruct the \sublist in $S_1$. 
    Replaying $\ReplicateDelete$ requests is as simple as first finding the corresponding node using the server ID and timestamp of the node, and then setting $isDeleted$ of the found node to be true. 
    However, for replaying $\ReplicateInsertAfter$ requests, there are two challenges : (1) there can be multiple $\ReplicateInsertAfter$ requests arrived and processed simultaneously, (2) these requests might not be processed in timestamp order of the nodes. Hence, the insert replays need to insert the nodes at its correct position in the \sublist, using its timestamp.
    Intuitively, $\ReplicateInsertAfter$ is handled as follows: First, observe that $\ReplicateInsertAfter$ identifies the previous node (call the node $A$) where $\InsertAfter$ was performed. 
    If two items, say $N_1$ and $N_2$, are inserted after node $A$, then the timestamp identifies which element was inserted first. Specifically, if timestamp of $N_1$ is less than that of $N_2$, that means that $N_1$ was inserted first and then $N_2$. This observation forces certain properties of the linked list observable through their timestamp (as explained in \Cref{TRAppendix}). This is used to replay the insertion on $S_2$. Note that this $\Replicate$ message does not impact client latency. 


    

\begin{algorithm}
\scriptsize
\DontPrintSemicolon
\caption{Updates in \TempReplMethod Protocol}\label{alg:TRAlgoUpdateOps}
\Fn{$\ItemRef$ \textbf{$\InsertAfter$}($\ItemRef$ prev, $\Key$ key)}{
\iffullalg
    prev$\PArrow$\StartCount.increment()\;
    \If{$prev.\StartCount\PArrow get() < 0$} {
        Send DelegateInsertAfter(prev.\newLocation,item, client) to prev.\newLocation.getSId()\;
        return \;
    }
    \Repeat{$\tt{rdcss}$(\&prev.isDeleted, false, \&prev.next, temp, newItem)}{
    \If{prev.isDeleted} {
        return ``Failed : Node not found"\;
    }
    
    $temp \gets prev.next$\;
    $newitem.next \gets temp$\;
    $newitem.\newLocation \gets prev.\newLocation$\;
    $newitem.key \gets key$\;
    $newItem.isDeleted \gets false$\;
    $newItem.\StartCount \gets prev.\StartCount $\;
    $newItem.\EndCount \gets prev.\EndCount $\;
    $newItem.\TS \gets \LC.IncrementAndGet()$\;\label{alg:TRAlgoUpdateTS}
    $newItem.sId \gets myId$\;
    }
    Send newItem to client\;
    \else
    Lines \ref{alg:AMAlgoUpdateInsertStartCount} to \ref{alg:AMInsertBeforeEndCount} of \Cref{alg:AMAlgoUpdate} and while creating newItem, fetch TS by incrementing logical clock of server and set sId as myId of server\; \label{alg:TRAlgoUpdateTS}
    \fi
    \If{$newItem.\newLocation.getSId() \neq myId$} {
        Send $\ReplicateInsertAfter$(prev.getItem(), newItem.getItem(), newItem) to newItem.$\newLocation$.getSId()\; \label{alg:ReplicateMsgInsert}
        return \;
    }
    $prev.\EndCount\PArrow increment()$\;

}
\Fn{$\tt{InsertReplayReceive}$($\ItemRef$ remoteRef, $\ItemRef$ localRef)} {
    $localRef.\newLocation \gets remoteRef$\;
    $localRef.\EndCount\PArrow increment()$\; \label{alg:ReplicateInsertAfterEndC}
    $\counterTemp \gets localRef.\EndCount\PArrow get()$\;
}
\Fn{\textbf{$\Delete$}($\ItemRef$ prev)}{
    \iffullalg
    prev$\PArrow$\StartCount.increment()\;
    \If{$prev\PArrow\StartCount < 0$} {
        Send DelegateDelete(prev.$\newLocation$, client) to prev.$\newLocation$.getSId()\;
        return \;
    }
    \Repeat{cas(\&prev.isDeleted, false, true)} {
        \If{prev.isDeleted} {
        return ``Failed : Node not found"\;
        }
    }
    Send ``Success" to client\;
    \else
    Lines \ref{alg:AMAlgoUpdateDeleteStartCount} to \ref{alg:AMDeleteBeforeEndCount} of \Cref{alg:AMAlgoUpdate}\;
    \fi
    \If{$prev.\newLocation.getSId() \neq myId$} {
        Send $\ReplicateDelete$(prev.getItem(), prev) to prev.\newLocation.getSId()\; \label{alg:ReplicateMsgDelete}
        return\;
    }
    $prev\PArrow\EndCount\PArrow increment()$\;
}

\Fn{$\tt{DeleteReplayReceive}$($\ItemRef$ localRef)} {
    $localRef.\EndCount\PArrow increment()$\; \label{alg:ReplicateDeleteEndC}
    $\counterTemp \gets localRef.\EndCount\PArrow get()$\;
}

\end{algorithm}

\begin{algorithm}
\scriptsize
\DontPrintSemicolon
\caption{$\Move$ in \TempReplMethod Protocol }\label{alg:TRAlgoLoadBalanceMove}

    \Fn {$\Move$($\ItemRef$ head, $\tt{Integer}$ newServerId)} {
    $remoteSH \gets \tt{MoveItem}(nil,head, newServerId, SH, nil)$\;\label{alg:TRAlgoMoveRecurse}
    $\SwitchLB$(head, newServerId, remoteSH)\;\label{alg:TRSwitchCall}
}

\Fn{$\ItemRef$ $\tt{MoveItem}$($\ItemRef$ remoteRef, $\ItemRef$ currRef, $\tt{Integer}$ newServerId, $\ItemRef$ SH, $\ItemRef$ remoteSH)} {
    \If{currRef.$\newLocation$.getSId() != myId}{
        $\tt{MoveItem}$(remoteRef, currRef.next, newServerId, SH, remoteSH)\;
    }
    $tempItem \gets currRef.getItem()$\;
    $response \gets Send Move(remoteRef, tempItem)$ to newServerId\;
    $currRef.\newLocation \gets response$\;
    \If{$tempItem.isDeleted \neq currRef.isDeleted$} {
        Send $\ReplicateDelete$ (currRef.getItem(), currRef) to currRef.$\newLocation$.getSId()\;        
    }
    \eIf{currRef.key = $\subtail$} {
        $\counterTemp \gets SH.\EndCount\PArrow get() + \Registry.get(SH).getOffset()$\;
        \While{!cas(\&currRef.\StartCount, \counterTemp, -$\infty$)\label{alg:TRAlgoLoadBalanceMoveCAS}}{
            \Comment*[r]{Wait for replays}\label{alg:TRMoveWaitforReplays}
        }
        return remoteSH\;

    } {
        \If{remoteSH = nil} {
            remoteSH = response\;
        }
        $\tt{MoveItem}$(response, currRef.next, newServerId, SH, remoteSH)\;
    }
}
\Fn{$\ItemRef$\ $\tt{MoveReceive}$($\ItemRef$ prev, $\tt{Item}$ item)}{
    $curr \gets prev$\;
    \eIf{item.key = $\subhead$} {
        Lines \ref{AMMoveReceiveSHStart} to \ref{AMMoveReceiveSHEnd} of \Cref{alg:AMAlgoLoadBalanceMove}\;
    } {
        \Repeat{$cas(\&currPrev.next, curr, newItem)$}{
            \Repeat{$curr.ts < prev.ts$} {
                $currPrev \gets curr$\;
                $curr \gets currPrev.next$\;
            }
            \eIf{item.key = $\subtail$} {
                $curr.next \gets item.next$\;
                return newItem\;
            } {
                \iffullalg
                $newItem \gets new ItemRef(item)$\;
                $newItem.\newLocation = nil$\;
                $newItem.\StartCount \gets currPrev.\StartCount $\;
                $newItem.\EndCount \gets currPrev.\EndCount $\;
                $newItem.next \gets curr$\;
                \else
                Create a $newitem$ by copying values from $currPrev$ and $item$ as done in Lines \ref{alg:BaseAlgoUpdateInsertTemp} to \ref{alg:BaseAlgoUpdateInsertLoopContentEnd} of \Cref{alg:BaseAlgUpdate}\;
                \fi
            }
        }
    }
    return newItem\;
}
\end{algorithm}
\begin{algorithm}
\scriptsize
\DontPrintSemicolon
\caption{Replays in \TempReplMethod Protocol}\label{alg:TRAlgoReplicateUpdateReceives}
\Fn{($\ItemRef$, $\ItemRef$) $\tt{{\ReplicateInsertAfter}Receive}$($\tt{Item}$ prevItem, $\tt{Item}$ item, $\ItemRef$ oldLocation)} {
    $prev \gets prevItem.\newLocation$\;\label{alg:TRFindPrevStart}
    \While{$prev.ts \neq prevItem.ts$ or $prev.sId \neq prevItem.sId$ 
    }{
        $prev \gets prev.next$\;
        \If{$prev.key = \subtail$}{
            $prev \gets prevItem.\newLocation$\;
        }
    } \label{alg:TRFindPrevEnd}
    $curr \gets prev$\;
    \Repeat{$cas(\&currPrev.next, curr, newItem)$}{
    \Repeat{$curr.ts < item.ts \label{TRInsertReplayConditionForInsertSpot}$ 
        } {
            $currPrev \gets curr$\; \label{TRInsertReplaycurrPrevFinding}
            $curr \gets currPrev.next$\; \label{TRInsertReplaycurrFinding}
        }
        \iffullalg
            $newItem \gets new ItemRef(item)$\;
            $newItem.\newLocation = nil$\;
            $newItem.\StartCount \gets currPrev.\StartCount $\;
            $newItem.\EndCount \gets currPrev.\EndCount $\;
            $newItem.next \gets curr$\;  
        \else
            Create a $newitem$ by copying values from $currPrev$ and $item$ as done in Lines \ref{alg:BaseAlgoUpdateInsertTemp} to \ref{alg:BaseAlgoUpdateInsertLoopContentEnd} of \Cref{alg:BaseAlgUpdate}\;
        \fi
        }    
    return newItem, oldLocation \Comment*[r]{Sent as InsertReplay Message}
}

\Fn{($\ItemRef$) $\tt{{\ReplicateDelete}Receive}$($\tt{Item}$ prevItem, $\ItemRef$ oldLocation)} {
    \iffullalg
    $prev \gets prevItem.\newLocation$\;
    \While{$prev = nil$ or $prev.ts \neq prevItem.ts$ or $prev.sId \neq prevItem.sId$
    }{
        $prev \gets prev.next$\;
        \If{$prev.key = \subtail$}{
            $prev \gets prevItem.\newLocation$\;
        }
    }
    \else
        Find $prev$ using Lines \ref{alg:TRFindPrevStart} to \ref{alg:TRFindPrevEnd}\ in \Cref{alg:TRAlgoReplicateUpdateReceives};
    \fi
    $prev.isDeleted \gets true$\;
    return oldLocation \Comment*[r]{Sent as DeleteReplay Message}
    
}

\end{algorithm}

\section{Discussion} \label{sec:Discussion}

In this section, we discuss some of the questions raised by the transformation of the linked list from a concurrent version to a distributed multi-machine list. 

\paragraph{When should $\Split$ and $\Move$ be invoked? } 
The decision of when these operations should be invoked would depend on specific needs of the clients/processes that utilize the list. To provide context,
$\Split$ is an $O(k)$ operation where $k$ is the size of the \sublist. Furthermore, if the \sublist grows too big, it will increase the duration of $\Split$ and $\Move$ operations. Additionally, having smaller \sublists permits the possibility of increasing concurrency in the $\Lookup$ operation. In this context, we recommend that the $\Split$ operation is invoked frequently so that each \sublist is small.



\paragraph{Comparison of \AbortingMoveMethod and \TempReplMethod }
    Both the protocols preserve the \textit{client-visible} linked list properties arising from its base implementation and in no circumstance, blocks client-visible operations on the list. 
    While the two protocols share the same $\Split$ method, the $\Move$ operations have their own benefits arising from their difference in implementation. The \AbortingMoveMethod protocol terminates the quickest when the $\Move$ is invoked on a \sublist at a period when the sublist is unlikely to undergo updates as explained in \Cref{AMAppendix}. Unlike \AbortingMoveMethod, the \TempReplMethod protocol is guaranteed to terminate for large \sublists on a momentary write-free period after the replay of the last concurrent update, as explained in \Cref{TRAppendix}. On the other hand,  \AbortingMoveMethod has the potential to move nodes in batches, thereby completing the overall $\Move$ of small \sublists faster than  \TempReplMethod . \TempReplMethod requires assigning $\newLocation$ for each node before moving the next node and hence moves only one at a time.

\paragraph{Can the transformation protocols be extended to sorted linked lists?}
Yes, and to search faster in a sorted linked list, our $\Split$ protocol can be modified to additionally store the minimum and maximum key values that will be stored in a particular \sublist in its $\Registry$.
To save space, the algorithm for a sorted list is provided in \Cref{sec:appendixsortedlistsection}.



\paragraph{How generic are the transformation protocols in order to apply them to other linked list data structures?}
The correctness proof of \AbortingMoveMethod protocol assumes a singly-linked list, and that any modification to a \sublist is contained within \StartCount and \EndCount increments of the corresponding \sublist. \TempReplMethod protocol is an illustrated optimization, where any such modification during a $\Move$ from servers $S1$ to $S2$ is reported to $S2$, to complete the \sublist reconstruction. As illustrated in \cite{Valois}, a context-aware unordered list can be extended to dictionary abstract data types such as sorted lists, hash tables, skip lists and binary search trees.

Due to reasons of space, we discuss some additional aspects such as 
\begin{inparaenum}
    \item Making $\Lookup$ return the first instance of a matching key 
    \item Distinction from graph partitioning
    \item Fault handling
\end{inparaenum}
 in \Cref{sec:DiscussionAppendix}.

\section{Related Work} \label{sec:RelatedWork}


Linked lists are fundamental data structures. A sequential linked list is easy to implement. However, a concurrent link list  is challenging to design, and hence, the concurrent linked list was studied extensively. 
There exists many lock-free implementations - \cite{Valois,harris2001pragmatic,Ruppert}; and a few wait-free implementations - \cite{zhangWaitFreeLinkedList,Timnat}. 
The first lock-free implementation using atomic CAS primitive (for non-distributed list) is due to Valois \cite{Valois}. This implementation is also context-aware due to the support of $\Next$ operation and insertions at current cursor position, and our approach can obtain a corresponding distributed implementation.

Some papers focused on optimizing the cost of operations
in addition to correctness and progress guarantees established in the aforementioned papers.
An example is a lock-free list implementation due to Fomitchev and Ruppert \cite{Ruppert}, with worst-case linear amortized cost. 
\cite{harris2001pragmatic} and \cite{Ruppert} 
perform de-linking of marked nodes at the time of read operations ($\tt{Contains}$). Our transformation operations assume such de-linking is paused during their executions. 
In spite of this, it is possible to transform the algorithm in \cite{harris2001pragmatic} to distributed linked-list as shown in the appendix. 
A concurrent lock-free doubly-linked list implementation is given by Attiya and Hillel \cite{Attiya}, using a double CAS operation. {This doubly-linked list does not have a read operation.} It is implemented using only CAS by Sundell and Tsigas \cite{Tsigas}. {Distributing \cite{Tsigas} can be a future work to augment our singly-linked list transormation to a context-aware doubly-linked list.}

The aforementioned implementations typically consider that the sequences of entries in the linked list reside on different chunks of memory. To take advantange of entries that are instead on the same chunk,
Braginsky and Petrank \cite{chunkedlist} developed a concurrent linked list that maintains in each memory chunk, a certain number of entries within a defined minimum and maximum. Their lock-free approach shares some similarities with our approach on partitioning the list across multiple machines.

Another well-celebrated direction is to use an universal construction approach, e.g. Herlihy \cite{herlihyMethodology}, to convert sequential implementations to wait-free concurrent implementations.  However, the universal construction adds overhead and direct implementations are preferred to avoid that overhead. 

Aforementioned literature is all on concurrent implementations on a single machine. 
In regard to distributed implementations on multiple machines, the literature is rare. The mostly related work is due to Abe and Yoshiday \cite{Distributeddoublylinkedlist} where a distributed doubly-linked list is constructed. The presented strategy uses conflict detection and sequence numbers (with some assumptions) to break symmetries.  
It also guarantees atomic insert and delete operations and non-blocking lookup operations.
Although this work is related, there are fundamental differences, such as the construction is not context-aware, and the algorithm is only obstruction-free. 
It also does not support load balancing large linked lists unlike our approach.

Zhang {\it et al.}\cite{zhangWaitFreeLinkedList} provided an unordered linked list where insertions always occurred at the head position of the linked list. 
This list is not context-aware, as the inserts always occur at the head position. Thus, the result of any operation supported by \cite{zhangWaitFreeLinkedList} (Insert, Delete or Contains) can be uniquely determined by the current state of the list without regard to the operations performed by the same client earlier. By contrast, the list in this paper is context-aware, as it supports $\InsertAfter$ (insert at given position in the list) and $\Next$. 

To the best of our knowledge, ours is the first distributed lock-free implementation of a context-aware linked list. 


\section{Conclusion and Future Work} \label{sec:Conclusion}

The ability to partition and distribute a data structure across several machines is highly valuable to increase the throughput. For this reason, approaches such as distributed key value stores and distributed hash tables are widespread in theory and practice. Unfortunately, these data structures are not context-aware, i.e., the result of an operation only depends upon the existing state of the data structure and not on previous operations performed by the same client.

To overcome this limitation, this paper focused on  techniques to transform a shared memory \textit{context-aware} data structure into a fully distributed one, without blocking client-visible operations. 
We focused on the list data structure as our case study. We started with an unordered linked list. Then, we presented two protocols that allowed the list to both be distributed in a manner that preserves the properties of the client-visible operations, and also never block the client-visible operations.
The first protocol ($\AbortingMoveMethod$) treated the load balancing task as a low priority task that could be abandoned if there were conflicting updates, where as, the second protocol ($\TempReplMethod$) 
had \textit{help} from client-visible operations 
to eliminate any need for the load balancing task to be abandoned, 
 without increasing the latency of the client-visible operations.

The properties of the protocols are such that, when there are no ongoing Switch operations,
the behavior of the list is retained, with additional capabilities, such as supporting parallel client requests to different servers.
During the $\SwitchLB$ operation, a client-visible operation may be forwarded to the new server. Other than that, the performance of client-visible operations remains unaffected.

There are several possible future work for this work.
While the paper illustrates the transformation on a singly-linked list, 
we believe that, in a future work, it can be extended to other data structures such as abstract dictionary data types(as in \cite{Valois}), a doubly-linked list\cite{Distributeddoublylinkedlist}, and external binary search trees \cite{concurrentBST}, where each tree can be split into sub-trees that are defined by their root nodes and be distributed across machines.

This work only focused on partitioning the list. It can be extended to a partitioned and \textit{replicated} list using approaches such as primary-backup\cite{primaryBackup}. Additionally, this work focused on preserving properties of client-visible operations of the list. For example, the $\InsertAfter$, $\Delete$ and $\Next$ (in a sublist) operations were linearizable in the list prior to the transformation and this property was preserved 
{even after the} transformation. Thus, another future work could be answering the natural question of whether weakened client-visible operation guarantees such as causal or eventual consistency, would provide a substantial increase in performance in distributed data structures. 


\newpage
{
\bibliographystyle{plain}
\bibliography{ref}
}

\begin{appendices}


\newpage

    \section{Additional Examples and Mathematical Formulation of of Context-awareness}\label{sec:additionalcontextaware}
    
    In this section, we give a mathematical definition of a data structure being context-aware. 
    This definition depends upon the requirement that the data structure is being accessed by multiple clients/threads. It is independent of whether the data structure can handle one operation at a time or whether multiple operations can execute concurrently. 

    A data structure can be viewed to be providing a set of transitions $F_s = f_1, f_2, \cdots$ (For example, a linked list provides operations $\InsertAfter(Ref(A), J)$, $\Next(Ref(B)$, etc.) to transition the state of the data structure from its current state $s$ to a new state (this new state can be the same state $s$).
    Note that the same client operation with different parameters (such as different keys), constitute to different elements in the set $F_s$. Further, assume only one operation can be executed by a client at a time. Given that the data structure is being accessed by multiple clients, consider client $c$ invoking an operation corresponding to transition $f_i$. We can view this invocation as a function $f_i$ to be of the form $f_i(s, h)$ where $s$ is the state of the data structure when the operation of $f_i$ is called, and $h$ is the (local) history of the transitions performed by the client $c$ on the data structure previously. Here, h provides \textit{context} of the client to the data structure, for the same operation and the same parameters. Function $f_i$ 
    returns $(s',v)$, where $s'$ is the updated state of the data structure and $v$ is the value returned to the client. 
    
    In other words, the result of the invocation of $f_i$ may depend upon the current state of the data structure and the history of the client invoking it. Further, given a history h of client c, a client maybe \textit{enabled} to perform only a subset of transitions $\prescript{h}{}{F_{s}} \subseteq F_s$ (for example, $\InsertAfter$ can only call with a reference that was returned in a previous operation). 
    
    We say that a data structure is context-aware iff there is at least one operation where the history (or context) of the client affects the result of the operation. In other words, a data structure is context-aware iff 
    \begin{align*}
        \exists \ (s,h_1,h_2) : \ &(\prescript{h_1}{}{F_s} \neq \prescript{h_2}{}{F_s}) \ \vee\\
        \ & (\exists f_i \in \prescript{h_1}{}{F_s} \  f_i(s,h_1) \neq f_i(s, h_2))
    \end{align*}

    As discussed in \Cref{sec:lockfreeconclist}, the most common way to identify if a data structure is context-aware is if it supports an operation that can only be invoked with a value returned from a previous data operation. In this case, the dependency on the history is visible syntactically, as a transition $f_i$ is enabled in one history but not in another, i.e., given a data structure state s, there exists two histories $h_1$, $h_2$ such that $f_i \in \prescript{h_1}{}{F_s} \wedge f_i \notin \prescript{h_2}{}{F_s}$. Thus, such data structures are context-aware. However, this dependency could be semantic only, i.e., the same transition $f_i$ is enabled in two histories ($f_i \in \prescript{h_1}{}{F_s} \wedge f_i \in \prescript{h_2}{}{F_s}$) but their result is different ($f_i(s,h_1) \neq f_i(s, h_2)$). A data structure whose operation transitions, in the presence of a client context, differ neither by syntax,
    nor by semantics,
    is not context-aware.

    As additional examples, a sorted linked list with insert(key), delete(key) and search(key) is not context-aware.
However, if it additionally supported $\Next()$ operation,
it would be context-aware.

Similarly, a tree that provides LeftChild and RightChild read operations for a given node is context-aware. A doubly-linked list that provides InsertBefore a node and Prev of a node as operations is also context-aware. However, a binary search tree that only provides set operations 
is not context-aware.

\section{Concurrent Shared Memory Linked List Appendix}\label{sec:concurrentListAppendix}
    This section provides information for the de-linking background process of \Cref{alg:BaseAlgUpdate} and additional explanation for its properties.
    
    \subsection{De-linking of nodes}\label{sec:delinkingAppendix}
    In addition to the logical delete in the $\Delete$ operation in \Cref{alg:BaseAlgUpdate}, we need to periodically de-link the logically deleted nodes out of the linked list, so that garbage collection techniques such as \cite{hazardPointers, debra} can be used to reclaim those memory locations. We do this as a separate process to make the transformation techniques easier to implement. 

    As illustrated in \Cref{alg:Delinking}, the periodic delinking process  traverses through each node in each \sublist while looking for a deleted node. If a node is deleted, the next pointer of its previous node is updated to point to the node after the deleted node using a CAS operation. If the CAS fails, it means the prev node had a concurrent $\InsertAfter$ operation taking place. Hence the next of a different node points to this deleted node. Hence we roll $curr$ back to this inserted node, and reach the deleted node later through the traversal. Upon a succesful CAS, operations that are currently on the deleted node, are now stuck with an outdated pointer to reach from curr.next. This is fixed by making it point to prev.next.

    In one execution of the delinking routine, any deleted node that is traversed is guaranteed to be delinked. A node that was inserted and deleted concurrently may not be traversed by a current delinking routine, but will be traversed and delinked as part of a future delinking routine. Hence every deleted node is eventually delinked. Periodicity of this routine controls the traversal overhead contributed by logically deleted nodes present in the linked list. 

    \subsection{Linearizability in the absence of $\Lookup$}\label{subsec:linearizabilityAppendix}
        The linearization point for $\InsertAfter$ (respectively, $\Delete$) is Line \ref{alg:BaseAlgoInsertCAS} (respectively, Line \ref{alg:BaseAlgoDeleteCAS}) of \Cref{alg:BaseAlgUpdate} when the \textit{RDCSS} operation (respectively, \textit{CAS}) completes successfully {and, hence, becomes the linearization point of \textit{RDCSS}\cite{MultiCAS}}. The linearization point of $\Next$ (in a sublist) is either just before any concurrent $\InsertAfter$ operations that took effect during $\Next$, or after Line \ref{alg:BaseAlgoNextIsDeleted}, when the operation finds an element that is not yet deleted, whichever of the two events happens earlier.

\begin{algorithm}[h]
\scriptsize
\DontPrintSemicolon
\caption{Delinking logically deleted nodes}\label{alg:Delinking}
\Comment{periodically called, but can also be called forcibly}
\Fn{DelinkDeletedNodes()}{
    \For{SL : $\Registry$} {
        $prev \gets SL.getSH()$\;
        $curr \gets prev.next$\;
        \Repeat{curr = $\subtail$} {
            \eIf{curr.isDeleted = True}{
                \If{CAS(\&prev.next, curr, curr.next)}{
                    $curr.next \gets prev.next$\;
                    \Comment{Enable garbage collection for curr}
                }
                $curr \gets prev.next$\;
            } {
                $prev \gets curr$\;
                $curr \gets curr.next$\;
            }   
        }
    }
}
\end{algorithm}

\section{$\AbortingMoveMethod$ Protocol Appendix}\label{AMAppendix}
    This section explains the correctness of $\AbortingMoveMethod$ Protocol through the following Lemmas: 
    \begin{lemma}\label{AMPropertiesLemma}
           Retaining properties of client-visible operations:
           \begin{enumerate}
               \item {In the absence of a $\Lookup$, $\InsertAfter$, $\Next$ and $\Delete$ are linearizable (regardless of whether a \sublist is undergoing any of the client-invisible operations that support the transformation).}
               \item $\Lookup$ holds the following property at all times: If an instance of key k is present throughout the $\Lookup$ operation, then that instance of k will be included in the response of $\Lookup(k)$. A $\Lookup$ may return ‘Not found’ if no item with the given key has persisted for the entire duration of the $\Lookup$.
           \end{enumerate}
    \end{lemma}
    \begin{proof}
        \begin{enumerate}[nosep, leftmargin=*]
            \item We prove the first claim by identifying linearization points for the client-visible operations under all circumstances as follows:
            \begin{itemize}[nosep, leftmargin=*]
                \item  \textbf{During a period when there is no on-going $\Split$, $\Move$ or $\SwitchLB$ operations on a \sublist:}
                        When there is no on-going $\Split$, $\Move$, or $\SwitchLB$, the nodes will always retain a non-negative \StartCount  as discussed in \Cref{subsec:AMAlgorithm}. Due to this, client-visible operations will fail \StartCount negative checks, and so, $\Delete$ and $\InsertAfter$ operations perform the operation identical to \Cref{alg:BaseAlgUpdate}. Hence they share the same linearization points as discussed in \Cref{subsec:LMLinearizability} (the CAS at Line \ref{alg:AMAlgoUpdateDeleteCAS} and RDCSS in Line \ref{alg:AMAlgoUpdateInsertCAS}, respectively, from \Cref{alg:AMAlgoUpdate}). 
                        For the $\Next$ operation, if the immediate $\Next$ node is a $\subtail$, then the request is delegated to the server handling requests for the next \sublist, as per the traversal. Thus, when there are no concurrent update operations on the node in the $\Next$ request, the linearization point for $\Next$ is when the execution of Line \ref{alg:AMAlgoReadNextElement} in some server results on a non-deleted node, that is not a $\subhead$ or a $\dummyNode$ (these can occur because of a concurrent $\Split$). While the $\Next$ operation can have delegations due to the next \sublist being in a different server, observe that this point of $\Next$ execution is a single atomic operation of checking $isDeleted$ value of the node about to be returned as the $\Next$ node. If there is a concurrent $\Delete$ on the node being returned, the same point of execution ($isDeleted$ check) will ensure that a deleted node is not returned. During concurrent $\InsertAfter$ on the node in the request of $\Next$, the linearization point has two cases : (1) The newly inserted node is returned - In this case linearization point is the $isDeleted$ check on the new node (2) An element that was formerly the immediate next element is returned - In this case, linearization point of the operation is just before the concurrent $\InsertAfter$ took effect. Thus, the linearization point of $\Next$ is the earlier of the two mentioned points of execution.   
                \item \textbf{During $\Split$:}
                When a \sublist is undergoing a $\Split$ operation, we have an assumption that nodes contained within it are not undergoing a $\Move$ operation. Thus, the nodes inside it would always retain a non-negative \StartCount throughout the course of the $\Split$ operation, as discussed in \Cref{subsec:AMAlgorithm}. Due to this, any node involved in the $\Split$ will fail \StartCount negative checks and so the $\Delete$, $\InsertAfter$ and $\Next$ operations maintain the same linearization points as if there is no ongoing $\Split$.
                \item \textbf{During $\Move$:}
                Until Line \ref{alg:AMAlgoLoadBalanceMoveCAS} of \Cref{alg:AMAlgoLoadBalanceMove} succeeds, \StartCount and \EndCount of the \sublist undergoing a $\Move$ remain non-negative. Hence, the linearization points upto this point is the same as client-visible operations on nodes that are not undergoing a $\Move$. Linearization point post Line \ref{alg:AMAlgoLoadBalanceMoveCAS} is discussed as part of $\SwitchLB$.
                \item \textbf{During $\SwitchLB$:}
                When ownership of a \sublist is being switched from server $S_1$ to $S_2$, the \StartCount and \EndCount of these nodes are negative in $S_1$. Hence the $\Next$, $\Delete$ and $\InsertAfter$ operations are delegated to $S_2$. Note that the nodes of the \sublist are stale in $S_1$ and so operations need to happen only in $S_2$. Hence the operations end up happening in $S_2$ with new \StartCount and \EndCount variables that are introduced as part of $\Move$ (Line \ref{AMMoveReceiveSHStart} to Line \ref{AMMoveReceiveSHEnd} of \Cref{alg:AMAlgoLoadBalanceMove}), which are already non-negative. $\Next$, $\Delete$ and $\InsertAfter$ in $S_2$ will therefore happen as if, there is no ongoing $\Split$, $\Move$, or $\Delete$, and thus retain the same linearization points.
            \end{itemize}
            \item $\Lookup$ Property : When the \StartCount is not negative, $\Lookup$ performs similar to \Cref{alg:BaseAlgUpdate}. When it is negative, we know that the \sublist in the current server has become stale. Thus, a request is sent to the server that currently owns the \sublist(Line \ref{alg:AMAlgoReadLookupDelegate2} of \Cref{alg:AMAlgoRead}), in order to fetch the nodes of a matching key. Thus, given a node that persisted throughout the duration of $\Lookup$ operation, a server will check if the node has a matching key either on its own, or through a delegation request that lets another server check on its behalf. Thus, for each $\Lookup$ operation, each node present at the invocation of $\Lookup$ in the list is traversed by some server, provided the node is not deleted concurrently during the $\Lookup$. This ensures that, if a node of a matching key persists during the entire duration of $\Lookup$, some server will return this node. Hence, if all servers return empty sets for the $\Lookup$ request, the client infers the collective response as a `Not found' message. Thus, the $\Lookup$ property described in \Cref{subsec:LMLinearizability} is preserved.
        \end{enumerate}
    \end{proof}

    \begin{definition}[Termination Condition]
        Termination condition of an operation is the circumstance that an operation requires in order to reach a point of execution, after which the number of remaining instructions is bounded.
    \end{definition}
    For example, for $\InsertAfter$ operations, the termination condition is a successful CAS operation. Until the CAS operation succeeds, a variable number of CAS operation failures could take place, and cause the lock-free $\InsertAfter$ operation to not terminate. Once the CAS operation completes, the number of remaining steps for the $\InsertAfter$ is bounded.   
    \begin{lemma}\label{AMTransOperationsLemma}
        The $\Split$, $\Move$ and $\SwitchLB$  operations satisfy their requirements described in \Cref{sec:Requirements} when their termination condition is met.
    \end{lemma}
    \begin{proof}
        \begin{itemize}[nosep, leftmargin=*]
    \item {\textbf{$\Split$: }}
        The $\Split$ operation in \Cref{alg:AMAlgoLoadBalanceSplit} can perform repeated steps because of concurrent interactions with client-visible operations through Lines \ref{alg:AMAlgoLoadBalanceSplitCAS} and \ref{alg:AMAlgoLoadBalanceSplitOffsetCompute} of \Cref{alg:AMAlgoLoadBalanceSplit}. Line \ref{alg:AMAlgoLoadBalanceSplitCAS} succeeds when $\Split$ wins the fair race that it partakes by executing \textit{RDCSS} instructions concurrently with $\InsertAfter$ operations on the same $node$. Line \ref{alg:AMAlgoLoadBalanceSplitOffsetCompute} succeeds when there are no concurrent update operations on the \sublist $SL$ undergoing the $\Split$, as explained in \Cref{subsec:AMAlgorithm}. When these termination conditions are met, the goal of the $\Split$ operation is to successfully achieve:
        \paragraph{Insertion of a $\subtail$ followed immediately by a $\subhead$ at the point of $\Split$}
        This is achieved by Line \ref{alg:AMAlgoLoadBalanceSplitCAS} inserting the 2 nodes into the \sublist. The second node is already a $\subhead$. The first node is inserted as a $\dummyNode$ but later marked as a $\subtail$ by Line \ref{alg:AMAlgoLoadBalanceSplitTruncate}.
        \paragraph{The resulting two \sublists have separate \StartCount , \EndCount and \offset parameters } 
        Separate \StartCount and \EndCount is achieved by creating \newStartCount and \newEndCount variables and assigning them for the right-half of the \sublist (nodes of $SL$ after the point of $\Split$) by Line \ref{alg:AMAlgoLoadBalanceSplitCounterUpdateFinished}. The \offsets are set after computing them through Line \ref{alg:AMAlgoLoadBalanceSplitOffsetCompute} loop condition, as explained in \Cref{subsec:AMAlgorithm}.
        \paragraph{The resulting two \sublists are individually present in the $\Registry$}
        This is achieved by Line \ref{alg:AMAlgoLoadBalanceSplitAddReg}, where the right-half \sublist is added to the $\Registry$. The left-half is already present in the $\Registry$, as it shares the same $\subhead$ as the \sublist that is being split ($SL$).
    \item {\textbf{$\Move$: }}
        The $\Move$ operation in \Cref{alg:AMAlgoLoadBalanceMove} can perform repeated steps because of concurrent interactions with client-visible operations, through the \textit{CAS} in Line \ref{alg:AMAlgoLoadBalanceMoveCAS}. This \textit{CAS} instruction on \StartCount can fail, if there was a concurrent $\InsertAfter$ or $\Delete$ operation taking place on the \sublist. This forces the $\Move$ to abort (Line \ref{alg:AMAlgoLoadBalanceMoveSLDelete}) and restart (Line \ref{alg:AMAlgoLoadBalanceMoveRetry}) during update operations on the \sublist. Hence the termination condition is to not have update operations until Line \ref{alg:AMAlgoLoadBalanceMoveCAS} takes effect. When this condition is met, $\Move$ iterates through all nodes of the \sublist and clones the \sublist on a new server, thereby reaching the stage where both servers have identical versions of the \sublist.
    \item {\textbf{$\SwitchLB$: }}
        The $\SwitchLB$ operation in \Cref{alg:AMAlgoLoadBalanceMove} does not require any special termination condition. It ensures the following:
        \setcounter{paragraph}{0}
        \paragraph{At least one server has ownership of the \sublist at all times}
            This is achieved by Line \ref{alg:AMAlgoLoadBalanceReceiveSwitchLine}, which in turn is achieved through Line \ref{alg:AMAlgoLoadBalance:switchRequest} where the request to the $\Registry$ addition is made. Once it is ensured that two servers own the \sublist, it is safely disowned by the moving server by Line \ref{alg:AMAlgoLoadBalance:switchRemoveReg}.
        \paragraph{The \sublist immediately before in traversal needs to point to the active version of this \sublist}
            This is achieved by Line \ref{alg:AMAlgoLoadBalance:switchPrecedingSubTail}, where the next pointer of the $\subtail$ of the previous \sublist is changed to point to the active \sublist.
        \paragraph{The stale \sublist is discarded only when clients can no longer hold reference to any of its nodes}
            This is achieved by the wait time of $\theta$ in Line \ref{alg:AMAlgoLoadBalance:switchWait}. As described in assumptions of \Cref{para:Assumptions}, the $\theta$ period is the maximum time upto which a client can receive and hold a reference to any node sent from the servers. The $\SwitchLB$ operation already marks the beginning of delegations of client-visible operations. Thus, clients holding references to stale nodes must have received them before $\SwitchLB$ operation had been initiated. Once $\SwitchLB$ operation has begun, it thus forces the maximum time that the stale references can be held by a client to be less than $\theta$, justifying $\theta$ to be a sufficient wait period in Line \ref{alg:AMAlgoLoadBalance:switchWait}, right before the sublist is marked for deletion.
        \end{itemize}

    \end{proof}

    \begin{theorem}
        The \AbortingMoveMethod protocol transforms the concurrent linked list in \Cref{alg:BaseAlgUpdate} into a distributed linked list while retaining its client-visible guarantees and supporting $\Split$, $\Move$ and $\SwitchLB$ transformation operations for the linked list.
    \end{theorem}
    \begin{proof}
        This follows directly from Lemma \ref{AMPropertiesLemma} and Lemma \ref{AMTransOperationsLemma}.
    \end{proof}
\section{$\TempReplMethod$ Protocol Appendix}\label{TRAppendix}
    This section explain the correctness of $\TempReplMethod$ Protocol through the following Lemmas:
    \begin{lemma}\label{TRPropertiesLemma}
        Retaining properties of client-visible operations:
        \begin{enumerate}
               \item {In the absence of a $\Lookup$, $\InsertAfter$, $\Next$ (in a sublist) and $\Delete$ are linearizable regardless of whether a \sublist is undergoing any of the client-invisible operations that support the transformation.}
               \item $\Lookup$ holds the following property at all times: If an instance of key k is present throughout the $\Lookup$ operation, then that instance of k will be included in the response of $\Lookup(k)$. A $\Lookup$ may return ‘Not found’ if no item with the given key has persisted for the entire duration of the $\Lookup$.
           \end{enumerate}
    \end{lemma}
    \begin{proof}
        The $\Delete$ and $\InsertAfter$ operations of $\TempReplMethod$ protocol return responses to clients before helping with $\Replicate$ messages even during a concurrent $\Move$. The pseudo code until returning a response to the client is the same in both $\AbortingMoveMethod$ and $\TempReplMethod$ protocols. The $\Next$ and $\Lookup$ are also the exact same operations in both protocols. Hence, the properties of client-visible operations directly follow from the same arguments discussed in \Cref{AMPropertiesLemma}.
    \end{proof}

    To prove the correctness of the individual transformation operations, we first prove the correctness of the Replay algorithm upon receiving any $\Replicate$ messages during an ongoing $\Move$. \Cref{subsec:TRAlgorithm} shares the intuition behind the replay. Upon receiving a request of the form $\ReplicateInsertAfter$($prevItem$, $item$, $oldLocation$), the operation is required to find the location of $prevItem$ using their unique (sID, ts) tuple and insert the $item$ at its appropriate place, and return the inserted location back to the requesting server. This server can now set the node matching the $oldLocation$ to have a $\newLocation$ provided from the replay message. While finding $prevItem$ and returning a node after insertion is straightforward, doing the insertion at the appropriate place requires some observations. For the sake of simplicity, oldLocation parameter is hidden in the discussion, as it does not affect the replay and is only used to complete assigning a newLocation on callback. We begin the proof by introducing the following terminology:
    \begin{definition}[Successor and Predecessor]
    If $\InsertAfter$(Ref(X), Y) has occurred, then Y is a \textit{successor} of X, and X is a \textit{predecessor} of Y. 
    \end{definition}

    \begin{definition}[Ancestor and Descendant]
        We say that X is an \textit{ancestor} of Y if and only if there is a sequence $X, X_1, X_2, \cdots, Y$ such that, each element in the list is a predecessor of the next element. In this case, we also say that $Y$ is a \textit{descendent} of $X$.
    \end{definition}

    \begin{definition}[Link notation $\rightarrow$]
        If the next pointer to a node A points to B, then we denote that a link of the form $A \rightarrow B$ exists. 
    \end{definition}

    \begin{definition}[Left and Right side of a node]
        A node A is said to be at the right side of a node B iff node A can be traversed from B by recursively calling the $\Next$ operation. Such a node B is said to be at the left side of node A.
    \end{definition}
    To prove the correctness of the Replay algorithm of \Cref{alg:TRAlgoReplicateUpdateReceives}, we utilize the following Lemmas:
    
    \begin{lemma}\label{LemmaForTimestamp}
        For any two requests $\InsertAfter$(A,B) and $\InsertAfter$(A,C) occur concurrently, the node that gets inserted first will have a lower timestamp and will be present further away from A.
    \end{lemma}
    \begin{proof}
        The lower timestamp clause of the lemma follows from the fact that, the insertion that lost the race in \textit{RDCSS} will re-increment the logical clock to get a timestamp higher than the previous values. The latter part of the observation follows from the fact that the node from the later $\InsertAfter$ will occupy the next pointer of A and will have its own next pointer pointing to the element that was inserted at A earlier.
    \end{proof}

    \begin{lemma}\label{LemmaforReplicate}
        If there is a $\Replicate$ message of the form $\ReplicateInsertAfter$(A,C), then $A.ts < C.ts$.
    \end{lemma}
    \begin{proof}
        This follows from the fact that A must have existed in the \sublist before $\InsertAfter$(Ref(A), C) was called.
    \end{proof}

    \begin{lemma} \label{LemmaForABeforeB}
        If there is a link of the form $A\rightarrow B$ in the \sublist that is being reconstructed, then one out of the following two has to be true: (1) node B was inserted at node A or descendants of A (and so has a timestamp greater than A); (2) B was inserted at a node preceding A (and so has a timestamp less than that of A).
    \end{lemma}
    \begin{proof}
        This follows from the fact that the entire linked list structure (excluding $isDeleted$ values) has been created by the use of $\InsertAfter$ and $\Split$ operations. Since $\Split$ only inserts special, client-invisible nodes ($\dummyNode$, $\subhead$ and $\subtail$), it is safe to argue that the nodes A and B must have been inserted into the list through $\InsertAfter$ operations. $\InsertAfter$ can only insert nodes to the right side of the node provided to it. Hence, if node B is appearing at the right of A, then either it was inserted at A or was inserted at an ancestor of A, at a time before A was inserted into the list. Thus, in the former case, $B.ts > A.ts$ and in the latter case, $B.ts < A.ts$.
    \end{proof}

    \begin{lemma}\label{LemmaB}
        If the replicated (under construction by replay) sublist in a server has a link of the form $A \rightarrow B$ and $A.ts > B.ts$, then any $\ReplicateInsertAfter$(A,C) must insert C in a place after A, but before B.
    \end{lemma}
    \begin{proof}
        This follows from Lemma \ref{LemmaForABeforeB}, by which B was not inserted at A in this case. Hence B is not a competing insertion at A, and C can be inserted just after A as the first replicated insertion at A.
    \end{proof}

    \begin{lemma}\label{LemmaD}
        If the replicated server has a link of the form $A\rightarrow B$, and $A.ts < B.ts$, then any $\ReplicateInsertAfter$(A,C) must traverse from A and find a node D such that $C.ts > D.ts$ and insert C just before D.
    \end{lemma}
    \begin{proof}
        From Lemma \ref{LemmaForABeforeB}, this conveys that B was inserted at A. This means C and B were possibly competing to insert at A, and so if $C.ts < B.ts$, C has to be present after B, according to Lemma \ref{LemmaForTimestamp}. After B, there could exist other descendants of A that were inserted in ways similar to B. Hence, the replay has to find the first node D after A, such that $C.ts > D.ts$, which will ensure either that D is not a descendant of A or that D is a descendant of A that was inserted into the list before C was inserted. Either case, C has to be inserted just before D.
    \end{proof}
    
    \begin{lemma}\label{TRReplayLemma}
        The Replay algorithm reconstructs the \sublist with the exact structure that is present in the source machine of the $\Move$. 
    \end{lemma}

    \begin{proof}
        Combining Lemmas \ref{LemmaB} and \ref{LemmaD}, the algorithm to replay a request of $\ReplicateInsertAfter$($prevItem$,$item$) condenses down to finding the first node($curr$ in Line \ref{TRInsertReplaycurrFinding} of \Cref{alg:TRAlgoReplicateUpdateReceives}), such that $curr.ts < item.ts$ (as in Line \ref{TRInsertReplayConditionForInsertSpot}). This will ensure that curr is either a node of type B in \ref{LemmaB} or a node of type D in \ref{LemmaD}. Hence, the insert replays reconstruct the same \sublist that the corresponding $\InsertAfter$ operations had created on the other server. 

    \end{proof}

    \begin{lemma}\label{TRTransOperationsLemma}
        The $\Split$, $\Move$ and $\SwitchLB$  operations satisfy their requirements described in \Cref{sec:Requirements} when their termination condition is met.
    \end{lemma}

    \begin{proof}
        \begin{itemize} 
        \item {\textbf{$\Split$: }} Correctness of $\Split$ operation follows from \AbortingMoveMethod method, as they share the same $\Split$ operation.
        \item {\textbf{$\Move$: }} 
            While the $\Move$ operation in $\TempReplMethod$ does not repeat the entire operation in a loop, a looped attempt on setting \StartCount to $-\infty$ is done through a \textit{CAS} operation at Line \ref{alg:TRAlgoLoadBalanceMoveCAS} of \Cref{alg:TRAlgoLoadBalanceMove}. The first success of this \textit{CAS} will indicate the termination of the $\Move$ operation (including ${\tt Replication}$). The condition for its termination is that no update operation should take place on a \sublist momentarily after the last \EndCount increment, at some point. When this happens, it indicates that the \EndCount and \StartCount had matched momentarily. Having the $\Replicate$ acknowledgements (the upon Receive functions in \Cref{alg:TRAlgoUpdateOps}) increment the \EndCount (by Lines \ref{alg:ReplicateInsertAfterEndC} and \ref{alg:ReplicateDeleteEndC} of \Cref{alg:TRAlgoUpdateOps}), $\Replicate$ message is processed before the the \EndCount gets incremented. Thus, when \StartCount and \EndCount are momentarily equal, the \sublist in both servers are identical, as guaranteed by the reconstruction (Lemma \ref{TRReplayLemma}). Hence $\Move$ completes successfully.      
        \item {\textbf{$\SwitchLB$}: } Correctness of $\SwitchLB$ operation follows from $\AbortingMoveMethod$ method, as they share the same $\SwitchLB$ operation.
    \end{itemize}
    \end{proof}

\begin{theorem}
        The \TempReplMethod protocol transforms the concurrent linked list in \Cref{alg:BaseAlgUpdate} into a distributed linked list while retaining its client-visible guarantees and supporting $\Split$, $\Move$ and $\SwitchLB$ transformation operations for the linked list.
    \end{theorem}
    \begin{proof}
        This follows directly from Lemma \ref{TRPropertiesLemma} and Lemma \ref{TRTransOperationsLemma}.
    \end{proof}
\section{Sorted Linked List} \label{sec:appendixsortedlistsection}

\subsection{Single Machine Implementation}
    We construct a sorted list with three client operations - $\Insert$(key), $\Search$(key) and $\Delete$(key), with each node having a unique key. Referring to the implementation in \Cref{alg:BaseSortedList}, $\Search$ and $\Delete$ both finds the position of the current key. $\Search$ returns this node and $\Delete$ deletes this node. We model this using our $\Lookup$(key) operation (\Cref{alg:BaseAlgUpdate}) for both operations, with $\Delete$ additionally utilizing $\Delete$(node) of \Cref{alg:BaseAlgUpdate}. 
    
    $\Insert$(key) requires a traversal that is different from just a $\Lookup$, as it requires not a matching key, but the highest key before the current key. We use $\Next$ from \Cref{alg:BaseAlgUpdate} to traverse and find such a node to insert after. Then we execute $\tt{TryInsertAfter}$, which is one round of loop execution of our $\InsertAfter$ operation in \Cref{alg:BaseAlgUpdate}. If the execution of RDCSS failed, we retry the $\Insert$ operation beginning from the traversal, instead of just redoing the insertAfter operation, to maintain sorted list and unique key properties.  

    Note that the three client operations ($\Insert$, $\Delete$ and $\Search$) are linearizable. $\Insert$ and $\Delete$ are linearizable at the RDCSS and CAS operations respectively. For $\Search$, there are two cases : 
    \begin{inparaenum}
        \item If $\Search$ returns a node of a matching key, then it is linearized at the $isDeleted$ check on the returned node
        \item If $\Search$ returns not found, then it is linearized at the last node whose next it traversed past.
    \end{inparaenum}

\begin{algorithm}[h]
\scriptsize
\DontPrintSemicolon
\caption{A lock-free unordered linked list}\label{alg:BaseSortedList}

\Struct{$\tt{Item}$} {
$\Key$ key,\;
$\ItemRef$ next,\;
$\tt{Boolean}$ isDeleted\;
}

\Fn{$\tt{Status}$ \  Insert($\Key$ key)}{
    SH $\gets$ $\Head$()\;
    $status \gets false$\;\label{alg:BaseSortedListInsertline2}
    \Repeat{status = true}{
        $prev \gets SH$\;
        $curr \gets \Next(prev)$\;
        \While{$curr.key \leq key\ or\ curr = SH$}{
            $prev \gets curr$\;
            $curr \gets \Next(prev)$\;
        }
        \If{prev.key = key} {
            return "Failed : Key already exists"\;
        }
        $status \gets TryInsertAfter(prev,key)$\;
    }
    return "Inserted successfully"\;\label{alg:BaseSortedListInsertEnd}
}

\Fn{$\tt{Status}$ \ $TryInsertAfter$($\ItemRef$ prev, $\Key$ key)}{
    \Comment{One execution of lines \ref{alg:BaseAlgoInsertCAS} - \ref{alg:BaseAlgoUpdateInsertAfterLoop}}
    return result of CAS\;
}

\Fn{$\tt{ItemRef}$ \ $Search$($\ItemRef$ prev, $\Key$ key)}{
    return Lookup(SH, key)\;
}

\Fn{$\tt{Delete}$ \ $\Delete$($\Key$ key)}{
    $node \gets Lookup(SH, key)$\;
    return $\Delete$(node)\;
}

\Comment[Lookup, Delete and Next are as implemented in \Cref{alg:BaseAlgUpdate}]    
\end{algorithm}

\subsection{Distributed Sorted Linked List Implementation}\label{sec:sortedlistdist}
    The key aspect of a sorted linked list that distinguishes it from an unordered linked list is the search or lookup operation for some kind of a key, that precedes every client operation. If our \sublists can help speed up this search by restricting the search to a \sublist instead of an entire list, then a distributed sorted list would have a better time complexity as the list grows, by periodically using $\Split$ to increase the number of \sublists, and thereby, putting a threshold on the number of elements per \sublist.

    In our $\Registry$, in addition to what we store in \Cref{subsec:AMDS}, a sublist will additionally be described by its $key_{min}$ and $key_{max}$ parameters. All nodes contained in a given sublist has unique keys and each key must be in the interval {$(key_{min}, key_{max}]$}. If there is only one sublist, then the interval should correspond to the key range for the entire list.

    With this construction, in \Cref{alg:TransformedSortedList}, we showcase how a client broadcasts all of its requests, and the server that owns the corresponding key to the client operation, would alone perform the operation by detecting it through $\tt{ClientRequestReceive}$. Once the server owning the sublist detects that it needs to perform the operation, it sends the registry entry for the sublist, along with the operation parameters, thereby reducing any search operations involved in any of the operations to be within this sublist. 

    The benefit of this approach is in time complexity of the search aspect of these operations. If the entire linked list has n nodes and the linked list is $\Split$ into m sublists, the search time is reduced from O(n) worst case time complexity to an amortized O(m + n/m).

\begin{algorithm}[h]
    \scriptsize
    \DontPrintSemicolon
    \caption{A lock-free unordered linked list}\label{alg:TransformedSortedList}
    $\Registry$: \textit{Each entry refers to a \sublist by storing its $\subhead$ pointer as key and a tuple value containing \{\StartCount, \EndCount, \offset , the subtail of its previous \sublist, $key_{min}$, $key_{max}$\}}\;\label{alg:SortedTransAlgoStructRegistry} 
    \Fn{$\ItemRef$ ClientRequestReceive(Operation op, Parameters param)}{
       $ key \gets param.getKey()$\;
       \For{SL : \Registry} {
            \If{ $key \in [SL.getKeyMin(), SL.getKeyMax())$}{
                $param \gets param \cup SL$\;
                op(param)\;
            } 
        }
   }

   \Fn{$\tt{Status}$ \  Insert(sublist SL, $\Key$ key)}{
        SH $\gets$ $SL.getSH()$\;
        Lines \ref{alg:BaseSortedListInsertline2} to \ref{alg:BaseSortedListInsertEnd}\;
    }

    \Fn{$\tt{ItemRef}$ \ $Search$(sublist SL, $\ItemRef$ prev, $\Key$ key)}{
        SH $\gets$ $SL.getSH()$\;
        return LookupSublist(SH, key)\;
    }

\Fn{$\tt{Delete}$ \ $\Delete$(sublist SL, $\Key$ key)}{
    SH $\gets$ $SL.getSH()$\;
    $node \gets LookupSublist(SH, key)$\;
    return $\Delete$(node)\;
}

\Comment{$\textbf{LookupSublist}$ and $\Next$ are as implemented in \Cref{alg:AMAlgoRead}. But $\Next$ allows \subhead to also be a valid returned item}

\Comment{$\Delete$(node) is as implemented in \Cref{alg:BaseAlgUpdate}}

\Comment{$\Delete$(SL, key) and $\Insert$(SL, key) transform with \StartCount and \EndCount based on the chosen transformation protocol}

\end{algorithm}

\subsection{Applying the transformation protocols}

Appendix Section \ref{sec:sortedlistdist} discussed the distributed sorted linked list when there is no ongoin $\Split$ or $\Move$ operation. Here, we discuss how to extend it to deal with these cases. 
Here we note that we can use the same approach as in Algorithms \ref{alg:AMAlgoRead} and \ref{alg:AMAlgoUpdate}. 
    Specifically, the technique of using the \StartCount and \EndCount variables to find write-free periods remains the same. The only addition the transformation protocols require is handling the change in key range whenever a $\Split$ is performed. 

    As illustrated in \Cref{alg:SortedListSplit}, after finding the $\offset$, we additionally need to compute the new key ranges for the split sublists before modifying the $\Registry$. This becomes possible from the observation that, the $\Insert$ operation searches for a (prev, curr) pair such that curr is a node of a greater key, and prev is the node just before curr. 
    
    When we perform $\Split$ at a target $\tt{node}$, we insert a ($\dummyNode$, $\subhead$) block at the $\tt{node}$. Let us call the target node as $Tnode$. Any $\Insert$ operations that take place after the block has been inserted, cannot insert a node after the $Tnode$ (in other words, it cannot insert a node between $Tnode$ and $\dummyNode$. It must place it after the newly inserted $\subhead$). This is because, in the (prev, curr) pair, prev cannot stay as the $Tnode$, as the corresponding curr would have been the $\subhead$ ($\Next$ operation skips past $\dummyNode$), which means the loop in $\Insert$ operation continues to find the $\Next$ node for curr, updating prev to be $\subhead$. This behavior makes $Tnode$ possess the highest key for the first(left) half of the sublist. Hence we store the key range for the left half as {$(oldkey_{min}, Tnode.key]$} and the right half as {$(Tnode.key, oldkey_{max}]$}.
\begin{algorithm}[h]
    \scriptsize
    \DontPrintSemicolon
    \caption{Modifying Split to support faster sorted key search}\label{alg:SortedListSplit}
    \Fn{($\ItemRef$,$\ItemRef$) \textbf{$\Split$}($\ItemRef$ SL, $\ItemRef$ node)} {
        \Comment{Execute upto Line \ref{alg:AMAlgoLoadBalanceSplitOffsetCompute} from \Cref{alg:AMAlgoLoadBalanceSplit}}
        $key \gets node.key()$\;
        $\Registry$.add(SH, \{\newStartCount, \newEndCount, $a_1$, ST, key, $\Registry$.get(SL).getKeyMax()\})\; \label{alg:SortedListAlgoLoadBalanceSplitAddReg}
        $\Registry$.get(SL).setOffset($a_2$)\; \label{alg:SortedListAlgoLoadBalanceSplitUpdateReg}
        $\Registry$.get(SL).setKeyMax(key)\;
        $ST.key \gets \subtail$\; \label{alg:SortedListAlgoLoadBalanceSplitTruncate}
        return $(ST,SH)$\;
    }

    \Comment{$\Move$ and $\SwitchLB$ remain the same in both protocols, except that on moving SH as first node, the $key_min$ and $key_max$ parameters are also copied .}
\end{algorithm}

With the above changes, we could support a distributed sorted linked list. We supported the operations $\Insert$, $\Delete$ and $\Search$, as these operations are traditionally considered in a sorted linked list implementation \cite{Valois}. Our description can additionally support the $\Next$ operation. If we make this client-visible (by doing a second $\Next$ call if $\subhead$ is returned in the first call of $\Next$), it will be a context-aware data structure. 


\section{Discussion Appendix}\label{sec:DiscussionAppendix}
\paragraph{Can $\Lookup$ be made to return the first instance of the key?}

Yes, and the intuition is as follows: we maintain an additional counter in a \sublist $\Registry$ to identify where the \sublist lies. This way, the client stub can identify which \sublist is traversed earlier in the original list. This counter would need to be changed during the $\Split$ operation. To simplify this task, we could maintain the counter as a range of values rather than a single value. {(For example, splitting a \sublist $SL_1$ should generate two \sublists of the form $SL_{1.1}$ and $SL_{1.2}$, splitting $SL_{1.1}$ further should generate $SL_{1.1.1}$ and $SL_{1.1.2}$ and so on)}. In this case, the $\Split$ operation will simply divide the range into two sub-parts. 

\paragraph{What is the difference between graph partitioning and list partitioning?}

Graph partitioning has been studied in detail in the literature (e.g,, \cite{A-RK2006IPDPS}). Generally, in this work, the main concern is performing tasks when there is an edge $(i, j)$ and $i$ and $j$ are stored on two different machines. They do not focus on adding or removing nodes from a graph while performing load balancing, as done in this work. 

\paragraph{Can there be support for crash tolerance?}
 Yes, by using a replicated state machine algorithm such as Raft\cite{raft} for any partition that needs crash tolerance. 

\end{appendices}
\end{document}